\documentclass[11pt,draftcls,peerreview,onecolumn]{IEEEtran}

\usepackage{cite}
\usepackage{subfigure}
\usepackage{url}
\usepackage{amsmath}
\usepackage{amssymb}
\usepackage{amssymb}
\usepackage{latexsym}
\usepackage{bm}
\usepackage[dvips]{graphics}
\usepackage{graphicx}
\usepackage{psfrag}

\newtheorem{definition}{Definition}

\newtheorem{theorem}{Theorem}
\newtheorem{lemma}[theorem]{Lemma}

\newtheorem{corollary}[theorem]{Corollary}
\newtheorem{proposition}{Proposition}
\newtheorem{example}{Example}

\begin{document}

\title{Feedback Capacity of the Compound Channel}


\author{Brooke~Shrader,~\IEEEmembership{Student Member, IEEE}
        and~Haim~Permuter,~\IEEEmembership{Student Member, IEEE}
}

\markboth{Submitted to IEEE Transactions on Information Theory,
2007.}{Shell \MakeLowercase{\textit{et al.}}: Feedback Capacity of
the Compound Finite State Channel}


\maketitle

\begin{abstract}
In this work we find the capacity of a compound finite-state channel
with time-invariant deterministic feedback. The model we consider
involves the use of fixed length block codes.
Our achievability result includes a proof of the existence of a
universal decoder for the family of finite-state channels with
feedback. As a consequence of our capacity result, we show that
feedback does not increase the capacity of the compound
Gilbert-Elliot channel. Additionally, we show that for a stationary
and uniformly ergodic Markovian channel, if the compound channel
capacity is zero without feedback then it is zero with feedback.
Finally, we use our result on the finite-state channel to show that
the feedback capacity of the memoryless compound channel is given by
$\inf_{\theta} \max_{Q_X} I(X;Y|\theta)$.
\end{abstract}


\section{Introduction}

The compound channel consists of a set of channels indexed by
$\theta\in \Theta$ with the same input and output alphabets but
different conditional probabilities. In the setting of the compound
channel only one actual channel $\theta$ is used in all
transmissions. The transmitter and the receiver know the family of
channels but they have no prior knowledge of which channel is
actually used. There is no distribution law on the family of
channels and the communication has to be reliable for all channels
in the family.

Blackwell et al. \cite{Blackwell59Compound} and independently
Wolfowitz \cite{Wolfowitz59} showed that the capacity of a compound
channel consisting of memoryless channels only, and without
feedback, is given by
\begin{equation} \label{eqn:CapacityMemorylessCompoundNoFB}
\max_{Q_X} \inf_{\theta } \mathcal I(Q_X;P_{Y|X,\theta}),
\end{equation}
where $Q_X(\cdot)$ denotes the input distribution to the channel,
$P_{Y|X,\theta}(\cdot|\cdot,\theta)$ denotes the conditional
probability of a memoryless channel indexed by $\theta$, and the
notation $\mathcal I(Q_X;P_{Y|X,\theta})$ denotes the mutual
information of channel $P_{Y|X,\theta}$ for the input distribution
$Q_X$, i.e.,
\begin{equation}
\mathcal I(Q_X;P_{Y |X, \theta }) \triangleq
\sum_{x,y}Q_X(x)P_{Y|X,\theta}(y|x,\theta)\ln
\frac{P_{Y|X,\theta}(y|x,\theta)}{\sum_{x'}Q_X(x')P_{Y|X,\theta}(y|x',\theta)}.
\label{eqn:mutual_info}
\end{equation}
The capacity in (\ref{eqn:CapacityMemorylessCompoundNoFB}) is in
general less than the capacity of every channel in the family.
Wolfowitz, who coined the term ``compound channel,'' showed that if
the transmitter knows the channel $\theta$ in use, then the capacity
is given by \cite[chapter 4]{Wolfowitz64}
\begin{equation}
\inf_{\theta } \max_{Q_X} \mathcal I(Q_X;P_{Y |X, \theta }) =
\inf_{\theta} C_\theta, \label{eqn:CFBDMC}
\end{equation}
where $C_\theta$ is the capacity of the channel indexed by $\theta$.
This shows that knowledge at the transmitter of the channel $\theta$
in use helps in that the infimum of the capacities of the channels
in the family can now be achieved. In the case that $\Theta$ is a
finite set, then it follows from Wolfowitz's result that
$\min_{\theta} C_{\theta}$ is the feedback capacity of the
memoryless compound channel, since the transmitter can use a
training sequence together with the feedback to estimate $\theta$
with high probability. In this paper we show that when $\Theta$ is
not limited to finite cardinality, the feedback capacity of the
memoryless compound channel is given by $\inf_{\theta} C_{\theta}$.
One might be tempted to think that for a compound channel with
memory, feedback provides a means to achieve the infimum of the
capacities of the channels in the family. However this is not
necessarily true, as we show in Example \ref{example:compoundGE},
which is taken from \cite{LapidothTelatar98} and applied to the
compound Gilbert-Elliot channel with feedback. That example is found
in Section \ref{section:compoundGE}.

A comprehensive review of the compound channel and its role in
communication is given by Lapidoth and Narayan
\cite{Lapidoth98Narayan}. Of specific interest in this paper are
compound channels with memory which are often used to model wireless
communication in the presence of fading
\cite{Gallager68,goldsmith96capacity,Mushkin89}. Lapidoth and
Telatar \cite{LapidothTelatar98} derived the following formula for
the compound channel capacity of the class of finite state channels
(FSC) when there is no feedback available at the transmitter.
\begin{equation}\label{eqn::capacity_no_feedback}
\lim_{n\to \infty} \max_{Q_{X^n}} \inf_{s_o,\theta} \frac{1}{n}
\mathcal I(Q_{X^n};P_{Y^n|X^n,s_0,\theta}),
\end{equation}
where $s_0$ denotes the initial state of the FSC, and
$Q_{X^n}(\cdot)$ and
$P_{Y^n|X^n,s_0,\theta}(\cdot|\cdot,s_0,\theta)$ denote the input
distribution and channel conditional probability for block length
$n$. Lapidoth and Telatar's achievability result makes use of a
universal decoder for the family of finite-state channels. The
existence of the universal decoder is proved by Feder and Lapidoth
in \cite{FederLapidoth98} by merging a finite number of
maximum-likelihood decoders, each tuned to a channel in the family
$\Theta$.

Throughout this paper we use the concepts of causal conditioning and
directed information which were introduced by Massey in
\cite{Massey90}. Kramer extended those concepts and used them in
\cite{Kramer03} to characterize the capacity of discrete memoryless
networks. Subsequently, three different proofs  -- Tatikonda and
Mitter \cite{Tatikonda00,Tatikonda06}, Permuter, Weissman and
Goldsmith \cite{Permuter06_feedback_submit} and Kim
\cite{Kim07_feedback} -- have shown that directed information and
causal conditioning are useful in characterizing the feedback
capacity of a point-to-point channel with memory. In particular,
this work uses results from \cite{Permuter06_feedback_submit} that
show that Gallager's \cite[ch. 4,5]{Gallager68} upper and lower
bound on capacity of a FSC can be generalized to the case that there
is a time-invariant deterministic feedback, $z_{i-1}=f(y_{i-1})$,
available at the encoder at time $i$.


In this paper we extend Lapidoth and Telatar's work for the case
that there is deterministic time-invariant feedback available at the
encoder by replacing the regular conditioning with the causal
conditioning. Then we use the feedback capacity theorem to study the
compound Gilbert-Elliot channel and the memoryless compound channel
and to specify a class of compound channels for which the capacity
is zero if and only if the feedback capacity is zero. The proof of
the feedback capacity of the FSC is found in Section
\ref{section:converse}, which describes the converse result, and
Section \ref{section:achievability}, where we prove achievability.
As a consequence of the capacity result, we show in Section
\ref{section:compoundGE} that feedback does not increase the
capacity of the compound Gilbert-Elliot channel. We next show in
Section \ref{section:IFF} that for a family of stationary and
uniformly ergodic Markovian channels, the capacity of the compound
channel is positive if and only if the feedback capacity of the
compound channel is positive. Finally, we return to the memoryless
compound channel in Section \ref{section:memorylesscompound} and
make use of our capacity result to provide a proof of the feedback
capacity. \footnote{Although Wolfowitz mentions the feedback problem
in discussing the memoryless compound channel \cite[ch.
4]{Wolfowitz64}, to the best of our knowledge, this result has not
been proved in any previous work.}

The notation we use throughout is as follows. A capital letter $X$
denotes a random variable and a lower-case letter, $x$, denotes a
realization of the random variable. Vectors are denoted using
subscripts and superscripts, $x^n = (x_1,\hdots,x_n)$ and
$x_i^n=(x_i,\hdots,x_n)$. We deal with discrete random variables
where a probability mass function on the channel input is denoted
$Q_{X^n}(x^n)=\mbox{Pr}(X^n = x^n)$ and
$P_{Y^n|X^n,\theta}(y^n|x^n,\theta)=\mbox{Pr}(Y^n=y^n|X^n=x^n,\theta)$
denotes a mass function on the channel output. When no confusion can
result,  we will omit subscripts from the probability functions,
i.e., $Q(x_i|x^{i-1},y^{i-1})$ will denote
$Q_{X_i|X^{i-1},Y^{i-1}}(x_i|x^{i-1},y^{i-1})$.

\section{Problem statement and main result}

The problem we consider is depicted in Figure \ref{fig:1}. A message
$W$ from the set $\{1,2,\hdots,e^{nR}\}$ is to be transmitted over a
compound finite state channel with time-invariant deterministic
feedback. The family $\Theta$ of finite state channels has a common
state space ${\cal S}$ and common finite input and output alphabets
given by ${\cal X}$ and ${\cal Y}$. For a given channel $\theta \in
\Theta$ the channel output at time $i$ is characterized by the
conditional probability
\begin{equation}
P(y_i, s_i | x_i, s_{i-1}, \theta), \quad y_i \in {\cal Y}, x_i \in
{\cal X}, s_i, s_{i-1} \in {\cal S}.
\end{equation}
which satisfies the condition $P(y_i, s_i | x^i, s^{i-1}, y^{i-1},
\theta) = P(y_i, s_i | x_i, s_{i-1}, \theta)$. The channel $\theta$
is in use over the sequence of $n$ channel inputs. The family
$\Theta$ of channels is known to both the encoder and decoder,
however, they do not have knowledge of the channel $\theta$ in use
before transmission begins.

\begin{figure}{
\psfrag{v1\r}[][][0.8]{$m$}\psfrag{w1\r}[][][0.8]{Message}
 \psfrag{u1\r}[][][0.8]{Encoder}
\psfrag{d1\r}[][][0.7]{$x_i(w,z^{i{-}1})$}
\psfrag{v2\r}[][][0.8]{$x_i$} \psfrag{w2\r}[][][0.8]{$$}
 \psfrag{u2\r}[][][0.8]{}
\psfrag{d2\r}[][][0.7]{$P(y_i,s_i|x_i,s_{i{-}1},\theta)$}
\psfrag{v3\r}[][][0.8]{$y_i$}
\psfrag{w3\r}[][][0.8]{$$}\psfrag{v1\r}[][][0.8]{$w$}
\psfrag{u3\r}[][][0.8]{Decoder} \psfrag{d3\r}[][][0.8]{$\hat
w(y^n)$} \psfrag{u5\r}[][][0.8]{Feedback Generator}
\psfrag{d5\r}[][][0.8]{$z_i(y_{i})$} \psfrag{v4\r}[][][0.8]{$z_i$}
\psfrag{w4\r}[][][0.8]{$$} \psfrag{u4\r}[][][0.8]{Unit Delay}
\psfrag{v5\r}[][][0.8]{$z_{i-1}$}
 \psfrag{w5\r}[][][0.8]{$$}
\psfrag{v6\r}[][][0.8]{$\hat w$}\psfrag{w6a\r}[][][0.8]{Estimated}
\psfrag{w6b\r}[][][0.8]{Message}
\centerline{\includegraphics[width=9cm]{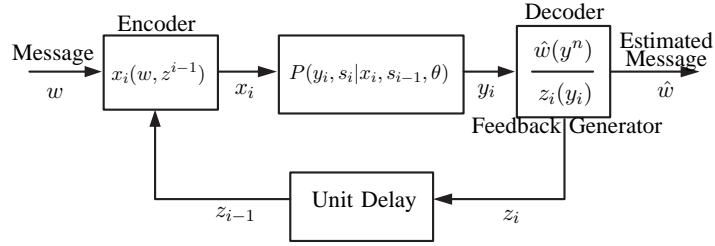}}
\caption{Compound finite state channel with feedback that is a
time-invariant deterministic function of the channel output.}
\label{fig:1} }\end{figure}

The message $W$ is encoded such that at time $i$ the codeword symbol
$X_i$ is a function of $W$ and the feedback sequence $Z^{i-1}$. For
notational convenience, we will refer to the input sequence
$X^i(W,Z^{i-1})$ as simply $X^i$. The feedback sequence is a
time-invariant deterministic function of the output $Y_i$ and is
available at the encoder with a single time unit delay. The function
performed on the channel output $Y_i$ to form the feedback $Z_i$ is
known to both the transmitter and receiver before communication
begins. The decoder operates over the sequence of channel outputs
$Y^n$ to form the message estimate $\hat{W}$.

For a given initial state $s_0 \in {\cal S}$ and channel $\theta \in
\Theta$, the channel causal conditioning distribution is given by
\begin{equation}
P(y^n || x^{n},s_0,\theta) \triangleq \prod_{i=1}^{n}
P(y_i|x^{i},y^{i-1},s_0,\theta).
\end{equation}
Additionally we will make use of Massey's directed information
\cite{Massey90}. When conditioned on the initial state and channel,
the directed information is given by
\begin{equation}
I(X^n \to Y^n | s_0,\theta) = \sum_{i=1}^{n}
I(Y_i;X^i|Y^{i-1},s_0,\theta).
\end{equation}
Our capacity result will involve a maximization of the directed
information over the input distribution $Q(x^n||z^{n-1})$ which is
defined as
\begin{equation}
Q(x^n||z^{n-1})\triangleq \prod_{i=1}^n Q(x_i|x^{i-1},z^{i-1}).
\end{equation}
We make use of some of the properties provided in
\cite{Massey90,Permuter06_feedback_submit} in our work, including
the following three which we restate for our problem setting.
\begin{enumerate}
\item $P(x^n,y^n|s_0,\theta)=Q(x^n||y^{n-1})P(y^{n}||x^n, s_0,\theta)$ \cite[eq. (3)]{Massey90} \cite[Lemma
1]{Permuter06_feedback_submit} \label{HaimChainRuleCausalCond}
\item $|I(X^n \to Y^n|\theta) - I(X^n \to Y^n|S,\theta)| \leq \log |{\cal S}|$, where random variable $S$ denotes the state of the finite-state channel \cite[Lemma 5]{Permuter06_feedback_submit}
\item From \cite[Lemma 6]{Permuter06_feedback_submit} \label{HaimDirectedInfoEquivI},
\begin{eqnarray*}
I(X^n \to Y^n | s_0,\theta ) & = & \mathcal
I(Q_{X^n||Y^{n-1}};P_{Y^n ||X^n,s_0,\theta})\\
 & = & \sum_{x^n,y^n}Q(x^n||y^{n-1}) P(y^n||x^n,s_0,\theta)\ln
\frac{P(y^n||x^n,s_0,\theta)}{\sum_{x'^{n}}Q(x'^{n})P(y^n||x'^{n},s_0,\theta)}
 \end{eqnarray*}
\end{enumerate}
Note that properties \ref{HaimChainRuleCausalCond}) and
\ref{HaimDirectedInfoEquivI}) hold since
$Q(x^n||y^{n-1},s_0,\theta)=Q(x^n||y^{n-1})$ for our feedback
setting, where it is assumed that the state $s_0$ is not available
at the encoder.

For a given initial state $s_0$ and channel $\theta$ the average
probability of error in decoding message $w$ is given by
\begin{equation*}
P_{e,w}(s_0,\theta) = \sum_{y^n \in {\cal Y}^n: \hat{w} \neq w}
P(y^n||x^n,s_0,\theta),
\end{equation*}
where $x^n$ is a function of the message $w$ and of the feedback
$z^{n-1}$. The average (over messages) error probability is denoted
$P_e(s_0,\theta)$, where $P_e(s_0,\theta)=1/e^{nR} \sum_{w}
P_{e,w}(s_0,\theta)$. We say that a rate $R$ is achievable for the
compound channel with feedback as shown in Figure \ref{fig:1}, if
for any $\epsilon >0$ there exists a code of fixed blocklength $n$
and rate $R$, i.e. $(n,e^{nR})$, such that $P_e(s_0,\theta) <
\epsilon$ for all $\theta \in \Theta$ and $s_0 \in {\cal S}$.
Equivalently, rate $R$ is achievable if there exists a sequence of
rate-$R$ codes such that
\begin{equation}
\lim_{n\to \infty} \sup_{s_0,\theta} P_e(s_0,\theta)=0.
\end{equation}
This definition of achievable rate is identical to that given in
previous work on the compound channel without feedback. A different
definition for the compound channel with feedback could also be
considered; for instance, in \cite{TchamkertenTelatar06}, the
authors consider codes of variable blocklength and define
achievability accordingly.

The capacity is defined as the supremum over all achievable rates
and is given in the following theorem.
\begin{theorem} \label{thrm:FBCapCompoundFSC}
The feedback capacity of the compound finite state channel is given
by
\begin{equation}
C = \lim_{n \rightarrow \infty} \max_{Q_{X^n||Z^{n-1}}} \inf_{s_0,
\theta} \frac{1}{n} I(X^n \to Y^n | s_0,\theta). \label{eqn:C}
\end{equation}
\end{theorem}
Theorem \ref{thrm:FBCapCompoundFSC} is proved in Section
\ref{section:converse}, which shows the existence of $C$ and proves
the converse, and Section \ref{section:achievability}, where
achievability is established.

\section{Existence of $C$ and the converse} \label{section:converse}

We first state the following proposition, which shows that the
capacity $C$ as defined in Theorem \ref{thrm:FBCapCompoundFSC}
exists. The proof is found in Appendix \ref{app:ExistenceProof}.

\begin{proposition} \label{prop:existenceC}
Let
\begin{equation}\label{eqn:C_nDefinition}
C_n = \max_{Q_{X^n||Z^{n-1}}} \inf_{s_0, \theta} \frac{1}{n} I(X^n
\to Y^n | s_0,\theta).
\end{equation}
Then $C_n$ is well defined and converges for $n \rightarrow \infty$.
In addition, let
\begin{equation}
\hat{C}_n = C_n - \frac{\log |{\cal S}|}{n}.
\end{equation}
Then
\begin{equation}
\lim_{n \rightarrow \infty} C_n = \sup_{n} \hat{C}_n
\end{equation}
\end{proposition}

To prove the converse in Theorem \ref{thrm:FBCapCompoundFSC}, we
assume a uniform distribution on the message set, for which
$H(W)=nR$. Since the message is independent of the channel
parameters $H(W)=H(W|s_0,\theta)$ and we apply Fano's inequality as
follows.
\begin{eqnarray*}
nR & = & H(W|s_0,\theta) \\
 & = & I(Y^n;W | s_0, \theta) + H(W|Y^n, s_0, \theta) \\
 & \leq & I(Y^n;W | s_0, \theta) + P_e(s_0,\theta)nR + 1 \\
 & = & H(Y^n | s_0, \theta) - H(Y^n|W, s_0, \theta) + P_e(s_0,\theta) nR + 1 \\
 & = & \sum_{i=1}^{n} H(Y_i|Y^{i-1}, s_0, \theta) - \sum_{i=1}^{n}
 H(Y_i|Y^{i-1}, W,s_0,\theta) + P_e(s_0,\theta) nR + 1 \\
& = & \sum_{i=1}^{n} H(Y_i|Y^{i-1}, s_0, \theta) - \sum_{i=1}^{n}
 H(Y_i|Y^{i-1}, W, X^i(W,Z^{i-1}(Y^{i-1})), s_0,\theta) + P_e(s_0,\theta) nR +
 1 \\
 & = & \sum_{i=1}^{n} H(Y_i|Y^{i-1}, s_0, \theta) - \sum_{i=1}^{n}
 H(Y_i|Y^{i-1}, X^i, s_0,\theta) + P_e(s_0,\theta) nR + 1 \\
 & = & \sum_{i=1}^{n} I(Y_i;X^i|Y^{i-1},s_0,\theta) + P_e(s_0,\theta) nR + 1 \\
 & = & I(X^n \to Y^n | s_0,\theta) + P_e(s_0,\theta) nR + 1
\end{eqnarray*}
For any code we have
\begin{equation}
I(X^n \to Y^n | s_0,\theta) \geq nR(1-P_e(s_0,\theta)) - 1
\end{equation}
and therefore
\begin{equation}
\inf_{s_0,\theta} I(X^n \to Y^n | s_0,\theta) \geq
nR(1-\sup_{s_0,\theta} P_e(s_0,\theta)) - 1.
\end{equation}
By combining the above statement with Proposition
\ref{prop:existenceC} we have
\begin{equation}
C \geq \hat{C}_n \geq R(1-\sup_{s_0,\theta} P_e(s_0,\theta)) -
\frac{1}{n} - \frac{\log|{\cal S}|}{n}.
\end{equation}
Then for a sequence of codes of rate $R$ with $\lim_{n\to \infty}
\sup_{s_0,\theta} P_e(s_0,\theta)=0$, this implies $R\leq C$.


\section{Achievability} \label{section:achievability}

Before proving achievability, we mention a simple case which follows
from previous results. If the set $\Theta$ has finite cardinality,
then achievability follows immediately from the results in
\cite[Theorem 14]{Permuter06_feedback_submit}, which are true for
any finite state channel with feedback. Hence, we can construct a
finite state channel where the augmented state is $(s,\theta)$ and
by assuming that the initial distribution is positive for all
$(s_0,\theta)$ then we get that for any $\theta \in \Theta,
|\Theta|<\infty$ and any $s_0\in \mathcal S $ the rate $R$ is
achievable if
\begin{equation}
R < \lim_{n \rightarrow \infty} \max_{Q_{X^n||Z^{n-1}}}  \min_{s_0,
\theta} \frac{1}{n} I(X^n \to Y^n | s_0,\theta).
\end{equation}

More work is needed in the achievability proof when the set $\Theta$
is not restricted to finite cardinality. This is outlined in the
following subsections in three steps. In the first step, we assume
that the decoder knows the channel $\theta$ in use and we show in
Theorem \ref{thrm:AchievabilityThetaKnown} that if $R < C$ and if
the decoder consists of a maximum-likelihood decoder, then there
exist codes for which the error probability decays uniformly over
the family $\Theta$ and exponentially in the blocklength. The codes
used in showing this result are codes of blocklength $Nm$ where each
sub-block of length $m$ is generated i.i.d. according to some
distribution. In the second step, we show in Lemma
\ref{lemma:TypeOnCodetree} that if instead the codes are chosen
uniformly and independently from a set of possible blocklength-$Nm$
codes, then the error probability still decays uniformly over
$\Theta$ and exponentially in the blocklength. In the third and
final step, we show in Theorem \ref{thrm:UnivDecoderForSeparable}
and Lemma \ref{lemma:CompoundFSCIsStronglySep} that for codes chosen
uniformly and independently from a set of blocklength-$Nm$ codes,
there exists a decoder that for every channel $\theta \in \Theta$
achieves the same error exponent as the maximum-likelihood decoder
tuned to $\theta$.

In the sections that follow, ${\cal P}({\cal X}^n || {\cal
Z}^{n-1})$ denotes the set of probability distributions on $X^n$
causally conditioned on $Z^{n-1}$.


\subsection{Achievability for a decoder tuned to $\theta$} We begin
by proving that if the decoder is tuned to the channel $\theta \in
\Theta$ in use, i.e., if the decoder knows the channel $\theta$ in
use, and if $R<C$ then the average error probability approaches
zero. This is proved through the use of random coding and maximum
likelihood (ML) decoding.

The encoding scheme consists of randomly generating a {\it
code-tree} for each message $w$, as shown in Figure
\ref{f_codetree}(b) for the case of binary feedback. A code-tree has
depth $n$ corresponding to the blocklength and level $i$ designates
a set of $|\mathcal Z|^{i-1}$ possible codeword symbols. One of the
$|\mathcal Z|^{i-1}$ symbols is chosen as the input $X_i$ according
to the feedback sequence $z^{i-1}$. The first codeword symbol is
generated as $X_1 \sim Q(x_1)$. The second codeword symbol is
generated by conditioning on the previous codeword symbol and on the
feedback, $X_2 \sim Q(x_2|x_1,z_1)$ for all possible values of
$z_1$. For instance, in the binary case, $|\mathcal Z|=2$, two
possible values (branches) of $X_2$ will be generated and the
transmitted codeword symbol will be selected from among these two
values according to the value of the feedback $Z_1$. Subsequent
codeword symbols are generated similarly, $X_i \sim Q(x_i|x^{i-1},
z^{i-1})$ for all possible $z^{i-1}$. For a given feedback sequence
$z^{n-1}$, the input distribution, corresponding to the distribution
on a path through the tree of depth $n$, is
\begin{equation}
Q(x^n ||z^{n-1}) = \prod_{i=1}^{n} Q(x_i|x^{i-1},z^{i-1})
\end{equation}

\begin{figure}[h]{
\psfrag{c1}[][][0.8]{$x_1=0$} \psfrag{a1}[][][0.8]{$$}
\psfrag{c2}[][][0.8]{$x_2=1$} \psfrag{a2}[][][0.8]{$i=1$}
\psfrag{c3}[][][0.8]{$x_3=1$} \psfrag{a3}[][][0.8]{$i=2$}
\psfrag{c4}[][][0.8]{$x_4=0$} \psfrag{a4}[][][0.8]{$i=3$}

\psfrag{d1}[][][0.8]{$x_1=0$} \psfrag{a1}[][][0.8]{$$}
\psfrag{d2}[][][0.8]{$x_2=1$} \psfrag{a1}[][][0.8]{$$}
\psfrag{d3}[][][0.8]{$x_2=1$} \psfrag{a1}[][][0.8]{$$}
\psfrag{d4}[][][0.8]{$x_3=0$} \psfrag{a1}[][][0.8]{$$}
\psfrag{d5}[][][0.8]{$x_3=1$} \psfrag{a1}[][][0.8]{$$}
\psfrag{d6}[][][0.8]{$x_3=1$} \psfrag{a1}[][][0.8]{$$}
\psfrag{d7}[][][0.8]{$x_3=1$} \psfrag{a1}[][][0.8]{$$}
\psfrag{da}[][][0.8]{$x_4=0$} \psfrag{db}[][][0.8]{$x_4=1$}
\psfrag{a1}[][][0.8]{$$} \psfrag{a5}[][][0.8]{$i=4$}

\psfrag{k0}[][][0.8]{$\;\;\;\;\;\;z_{i-1}=0$}
\psfrag{k1}[][][0.8]{$\;\;\;\;\;\;z_{i-1}=1$}
\psfrag{k01}[][][0.8]{$\:\:\:\:\:\:\:\:\:\:\:\:$ (no feedback)}

\psfrag{codeword}[][][0.8]{(a) codeword (no feedback)}
\psfrag{code-tree}[][][1]{(b) code-tree } \psfrag{con
code-tree}[][][1]{(c) concatenated code-tree}

\centerline{\includegraphics[width=16cm]{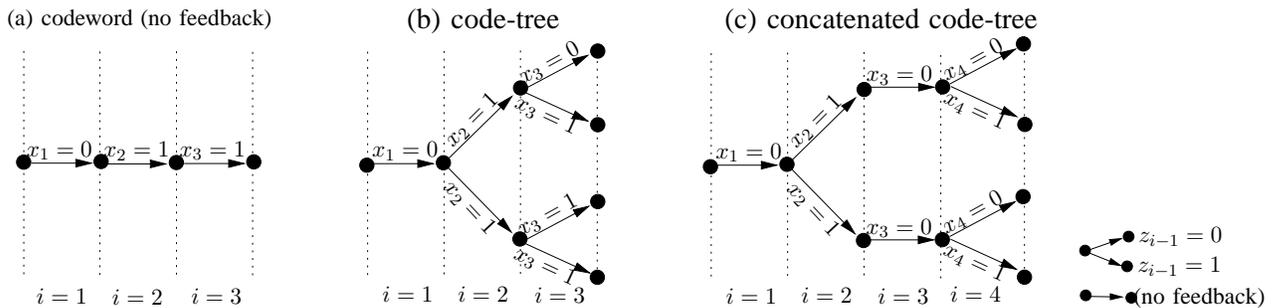}}
\caption{Illustration of coding scheme for (a) setting without
feedback, (b) setting with binary feedback as used in
\cite{Permuter06_feedback_submit} and (c) a code-tree that was
created by concatenating smaller code-trees. In the case of no
feedback each message is mapped to a codeword, and in the case of
feedback each message is mapped to a code-tree. The third scheme is
a code-tree of depth 4 created by concatenating two trees of depth
2. } \label{f_codetree} }
\end{figure}

A code-tree of depth $n$ is a vector of $D(n)$ symbols, where
\begin{equation}
D(n) \triangleq \sum_{i=1}^{n} |\mathcal Z|^{i-1} = \frac{|\mathcal
Z|^n -1}{|\mathcal Z|-1}, \label{eqn:CodeTreeNumSymbols}
\end{equation}
and each element in the vector takes value from the alphabet
$\mathcal X$. We denote a random code-tree by $A^{D(n)}$ and a
realization of the random code-tree by $a^{D(n)}$. The probability
of a tree $a^{D(n)} \in \mathcal X^{D(n)}$ is uniquely determined by
$Q_{X^n||Z^{n-1}}(\cdot||\cdot) \in {\cal P}({\cal X}^n || {\cal
Z}^{n-1})$. For instance, consider the case of binary feedback,
$\mathcal Z = \{0,1\}$, and a tree of depth $n=2$, for which
$D(n)=3$. A code-tree is a vector $a^{3}=(x_1,x_{21},x_{22})$ where
$x_1$ is the symbol sent at time $i=1$, $x_{21}$ is the symbol sent
at time $i=2$ for feedback $z_1=0$, and $x_{22}$ is the symbol sent
at time $i=2$ for feedback $z_1=1$. Then
\begin{equation}
\mbox{Pr}(A^3 = a^3) = Q(x_1)Q(x_{21}|x_1,z_1=0)Q(x_{22}|x_1,z_1=1)
\end{equation}
which is uniquely determined by $Q_{X^2||Z_1}(\cdot||\cdot)$. In
general, for a code-tree of depth $n$, the following holds.
\begin{equation}
\sum_{a^{D(n)} \in \mathcal X^{D(n)}} \mbox{Pr}(A^{D(n)} = a^{D(n)})
= 1
\end{equation}

A code-tree for each message $w$ is randomly generated, and for each
message $w$ and feedback sequence $z^{n-1}$ the codeword
$x^n(w,z^{n-1})$ is unique. The decoder is made aware of the
code-trees for all messages. Assuming that the ML decoder knows the
channel $\theta$ in use, it estimates the message as follows.
\begin{equation}
\hat{w} = \arg \max_{w} P(y^n|w,\theta)
\end{equation}
As shown in \cite{Permuter06_feedback_submit}, since $x^i$ is
uniquely determined by $w$ and $z^{i-1}$ and since $z^i$ is a
deterministic function of $y^i$, we have the equivalence
\begin{equation}
P(y^n|w,\theta) = P(y^n||x^n(w,z^{n-1}),\theta)
\end{equation}
so the ML decoder can be described as
\begin{equation}
\hat{w}=\arg \max_{w} P(y^n||x^n(w,z^{n-1}),\theta).
\end{equation}
Let $P_{e}^{n}(s_0,\theta)$ denote the average (over messages) error
probability incurred when a code of blocklength $n$ is used over
channel $\theta$ with initial state $s_0$. The following theorem
bounds the error probability uniformly in $(s_0,\theta)$ when the
decoder knows the channel $\theta \in \Theta$ in use. The theorem is
proved in Appendix \ref{app:AchievabilityThetaKnown}.

\begin{theorem} \label{thrm:AchievabilityThetaKnown}
For a compound FSC with initial state $s_0 \in {\cal S}$, input
alphabet ${\cal X}$, and output alphabet ${\cal Y}$, assuming that
the decoder knows the channel $\theta$ in use, then there exists a
code of rate $R$ and blocklength $Nm$, where $N \geq 1$ and $m$ is
chosen such that $\hat{C}_m \geq R + \epsilon$, for which the error
probability $P_{e}^{Nm}(s_0,\theta)$ of the ML decoder satisfies
\begin{equation}
P_{e}^{Nm}(s_0,\theta) \leq  |{\cal S}|
\exp(-Nm\beta(\epsilon,m,|{\cal Y}|))
\end{equation}
for any $\theta \in \Theta$, where
\begin{equation}
\beta(\epsilon,m,|{\cal Y}|) = \begin{cases} m \epsilon^2/(2
\log(e|{\cal Y}|^m)^2) & \epsilon < \frac{1}{m}  (\log(e|{\cal
Y}|^m) )^2
\\ \epsilon - \frac{1}{2m} \left( \log(e|{\cal Y}|^m) \right)^2 &
\mbox{otherwise}. \end{cases}
\end{equation}
\end{theorem}

The result in Theorem \ref{thrm:AchievabilityThetaKnown} is shown by
the use of a randomly-generated code-tree of depth $Nm$ for each
message $w$. For every feedback sequence $z^{Nm-1}$, the
corresponding path in the code-tree is generated by the input
distribution $Q_{X^{Nm}||Z^{Nm-1}}(\cdot||\cdot) \in {\cal P}({\cal
X}^{Nm}||{\cal Z}^{Nm-1})$ given by
\begin{multline}
Q(x^{Nm}||z^{Nm-1})  =  Q_m^*(x_1^{m}||z_1^{m-1}) \times
Q_m^*(x_{m+1}^{2m}||z_{m+1}^{2m-1}) \times \hdots \times
Q_m^*(x_{(N-1)m+1}^{Nm}||z_{(N-1)m+1}^{Nm-1}) \\
\forall x^{Nm} \in \mathcal X^{Nm}, z^{Nm-1} \in \mathcal Z^{Nm-1}
\end{multline}
where $Q_m^*$ is the distribution that achieves the supremum in
$\hat{C}_m$. The random codebook ${\cal C}$ used in proving Theorem
\ref{thrm:AchievabilityThetaKnown} consists of $e^{NR}$ code-trees.
Each code-tree in the codebook is a concatenated code-tree with
depth $Nm$ consisting of $N$ code-trees, each of depth $m$. For a
given feedback sequence $z^{Nm-1}$ (corresponding to a certain path
in the concatenated code-tree) the codeword is generated by
$Q_{X^{Nm}||Z^{Nm-1}}(\cdot||\cdot)$. An example of a concatenated
code-tree is found in Figure \ref{f_codetree}(c).

\subsection{Achievability for codewords chosen uniformly over a set}

In this subsection we show that the result in Theorem
\ref{thrm:AchievabilityThetaKnown} implies that the error
probability can be similarly bounded when codewords are chosen
uniformly over a set. In other words, we convert the random coding
exponent given in Theorem \ref{thrm:AchievabilityThetaKnown}, where
it is assumed that the codebook consists of concatenated code-trees
of depth $Nm$ in which each sub-tree of depth $m$ is generated
i.i.d. according to $Q_m^*$, to a new random coding exponent for
which the concatenated code-trees in the codebook are chosen
uniformly from a set of concatenated code-trees. This alternate type
of random coding, where the concatenated code-trees are chosen
uniformly from a set, is the coding approach subsequently used to
prove the existence of a universal decoder.

We first introduce the notion of types on code-trees. Let
$a^{ND(m)}$ denote the concatenation of $N$ depth-$m$ code-trees
$a^{D(m)}$, where $D(m)$ is defined in
(\ref{eqn:CodeTreeNumSymbols}) and $a^{ND(m)} \in \mathcal
X^{ND(m)}$. The type (or empirical probability distribution) of a
concatenated code-tree $a^{ND(m)}$ is the relative proportion of
occurrences of each code-tree $a^{D(m)} \in \mathcal X^{D(m)}$.
Equivalently, $N$ multiplied by the type of $a^{ND(m)}$ indicates
the number of times each depth-$m$ code-tree from the set ${\cal
X}^{D(m)}$ occurs in the concatenated code-tree $a^{ND(m)}$. Let
${\cal P}_N({\cal X}^{D(m)})$ denote the set of types of
concatenated code-trees of depth $Nm$.


Let $P_e(n,R,Q,P)$ denote the average probability of error incurred
when a code-tree of depth $n$ and rate $R$ drawn according to a
distribution $Q \in {\cal P}({\cal X}^{n}||{\cal Z}^{n-1})$ is used
over the channel $P$. We now prove the following result.

\begin{lemma} \label{lemma:TypeOnCodetree}
Given $Q_m \in {\cal P}({\cal X}^m||{\cal Z}^{m-1})$,
let $Q_{Nm} \in {\cal P}({\cal X}^{Nm}||{\cal Z}^{Nm-1})$ denote the
distribution given by the N-fold product of $Q_m$, i.e.,
\begin{equation}
Q_{Nm}(x^{Nm}||z^{Nm-1}) = \prod_{i=1}^{N}
Q_m(x_{(i-1)m+1}^{im}||z_{(i-1)m+1}^{im-1}), \quad \forall x^{Nm}
\in \mathcal X^{Nm}, z^{Nm-1} \in \mathcal Z^{Nm-1}
\end{equation}
For a given type $\hat{Q}_{Nm} \in {\cal P}_N({\cal X}^{D(m)})$, let
$\overline{Q}_{Nm} \in {\cal P}({\cal X}^{Nm}||{\cal Z}^{Nm-1})$
denote the distribution that is uniform over the set of concatenated
code-trees of type $\hat{Q}_{Nm}$. For every distribution $Q_m \in
{\cal P}({\cal X}^m||{\cal Z}^{m-1})$ there exists a type
$\hat{Q}_{Nm} \in {\cal P}_N({\cal X}^{D(m)})$ whose choice depends
on $Q_m$ and $N$ but not on $P$ such that
\begin{equation}
P_e(Nm,R,\overline{Q}_{Nm},P) \leq \exp(2Nm\delta(N,m,|{\cal Z}|))
P_e(Nm,R+m\delta(N,m,|{\cal Z}|),Q_{Nm},P)
\end{equation}
for all $P$, where $\delta(N,m,|{\cal Z}|)=|{\cal X}|^{D(m)}
\log(N+1)/Nm$ tends to 0 as $N \rightarrow \infty$.
\end{lemma}

\begin{proof}
The proof follows the approach of \cite[Lemma 3]{LapidothTelatar98}
except that our codebook consists of code-trees rather than
codewords; we include this proof for completeness in describing the
notion of types on code-trees. Given a codebook ${\cal C}$ of rate
$R+m\delta(N,m,|{\cal Z}|)$ chosen according to $Q_{Nm}$, we can
construct a sub-code ${\cal C}'$ of rate $R$ in the following way.
Let $Q'$ denote the type with the highest occurrence in ${\cal C}$.
The number of types in ${\cal C}$ is upper bounded by $(N+1)^{|{\cal
X}|^{D(m)}} = \exp(Nm\delta(N,m,|{\cal Z}|))$, so the number of
concatenated code-trees of type $Q'$ is lower bounded by
$\exp(N(R+m\delta(N,m,|{\cal Z}|)))/\exp(Nm\delta(N,m,|{\cal Z}|)) =
\exp(NR)$. We construct the code ${\cal C}'$ by picking the first
$e^{NR}$ concatenated code-trees of type $Q'$. Since ${\cal C}'$ is
a sub-code of ${\cal C}$, its average probability of error is upper
bounded by the average probability of error of ${\cal C}$ times
$|{\cal C}|/|{\cal C}'| = \exp(Nm\delta(N,m,|{\cal Z}|))$.

Conditioned on $Q'$, the codewords in ${\cal C}'$ are mutually
independent and uniformly distributed over a set of concatenated
code-trees of type $Q'$. Since ${\cal C}$ is a random code, the type
$Q'$ is also random, and let $\pi$ denote the distribution of $Q'$.
Pick a realization of the type $Q'$, denoted $\hat{Q}_{Nm}$, that
satisfies $\pi(\hat{Q}_{Nm}) \geq \exp(-Nm\delta(N,m,|{\cal Z}|))$.
(This is possible since the number of types is upper bounded by
$\exp(Nm\delta(N,m,|{\cal Z}|))$.) Then
\begin{eqnarray}
\pi(\hat{Q}_{Nm}) P_e(Nm,R,\overline{Q}_{Nm},P) \!&\! \leq &\!\!
\sum_{Q'}
\pi(Q') P_e(Nm,R,Q',P) \\
 \!&\! \leq &\!\! \exp(Nm\delta(N,m,|{\cal Z}|)) P_e(Nm, R+m\delta(N,m,|{\cal Z}|), Q_{Nm},P)
\end{eqnarray}
and
\begin{eqnarray}
P_e(Nm,R,\overline{Q}_{Nm},P) & \leq &
\frac{\exp(Nm\delta(N,m,|{\cal Z}|))}{\pi(\hat{Q}_{Nm})} P_e(Nm,
R+m\delta(N,m,|{\cal Z}|), Q_{Nm},P) \\
 & \leq & \exp(2Nm\delta(N,m,|{\cal Z}|)) P_e(Nm,
R+m\delta(N,m,|{\cal Z}|), Q_{Nm},P)
\end{eqnarray}
\end{proof}

Combining this result with Theorem
\ref{thrm:AchievabilityThetaKnown}, we have that there exists a type
$\hat{Q}_{Nm} \in {\cal P}_N({\cal X}^{D(m)})$ such that when the
codewords are chosen uniformly from the type class of
$\hat{Q}_{Nm}$, given by the distribution $\overline{Q}_{Nm}$, the
average probability of error is bounded as
\begin{eqnarray}
P_e(Nm,R,\overline{Q}_{Nm},P) & \!\!\leq\! &
\exp(2Nm\delta(N,m,|{\cal Z}|)) |{\cal
S}| \exp(-Nm\beta(\epsilon {-} m\delta(N,m,|{\cal Z}|)/2, m, |{\cal Y}|)) \\
 & \!\! = \! & |{\cal S}| \exp\left\{-Nm \left[ \beta\left( \epsilon {-} \frac{1}{2} m\delta(N,m,|{\cal Z}|), m, |{\cal Y}|\right) -2\delta(N,m,|{\cal Z}|)
 \right]
 \right\}
\end{eqnarray}


It is then possible to choose $N_0$ such that for all $N > N_0$,
\begin{equation}
\frac{1}{2}|{\cal X}|^{D(m)} \frac{\log(N+1)}{N} <
\frac{\epsilon}{2}
\end{equation}
and
\begin{equation}
2|{\cal X}|^{D(m)} \frac{\log(N+1)}{Nm} < \frac{1}{2}
\beta\left(\frac{\epsilon}{2},m,|{\cal Y}| \right)
\end{equation}
which implies that the probability of error is bounded as
\begin{equation}
P_e(Nm,R,\overline{Q}_{Nm},P) \leq |{\cal S}| \exp\left(-Nm
\frac{1}{2} \beta\left(\frac{\epsilon}{2},m,|{\cal Y}| \right)
\right)
\end{equation}

\subsection{Existence of a universal decoder}

We next show that when a codebook is constructed by choosing
code-trees uniformly from a set, there exists a universal decoder
for the family of finite-state channels with feedback. This result
is shown in the following four steps.
\begin{itemize}
\item We define the notion of a strongly separable family $\Theta$ of
channels given by the causal conditioning distribution. The notion
of strong separability means that the family is well-approximated by
a finite subset of the channels in $\Theta$.
\item We prove that for strongly separable $\Theta$ and code-trees
chosen uniformly from a set, there exists a universal decoder.
\item We describe the universal decoder which ``merges'' the ML
decoders tuned to a finite subset of the channels in $\Theta$.
\item We show that the family of finite-state channels given by the
causal conditioning distribution is a strongly separable family.
\end{itemize}
Our approach follows precisely the approach of Feder and Lapidoth
\cite{FederLapidoth98} except that our codebook consists of
concatenated code-trees (rather than codewords) and our channel is
given by the causal conditioning distribution.

Let $a^{ND(m)}$ denote a concatenated code-tree of depth $Nm$,
$a^{ND(m)} \in {\cal X}^{ND(m)}$ where $D(m)=(|{\cal Z}|^m -
1)/(|{\cal Z}| - 1)$, and let $B_{Nm}$ denote a set of such
code-trees, $B_{Nm} \subseteq {\cal X}^{ND(m)}$. As described in
Lemma \ref{lemma:TypeOnCodetree}, $B_{Nm}$ will be the set of
code-trees of type $\hat{Q}_{Nm} \in {\cal P}_N({\cal X}^{D(m)})$
and the code-tree for each message will be chosen uniformly from
this set, i.e. $\overline{Q}_{Nm}(a^{ND(m)}) = 1/|B_{Nm}|$ for any
$a^{ND(m)} \in B_{Nm}$. As described below, for a given output
sequence $y^{Nm}$, ML decoding will correspond to comparing the
functions $P_{\theta}(y^{Nm}|a^{ND(m)})$, $a^{ND(m)} \in B_{Nm}$.
Note that comparing the functions $P_{\theta}(y^{Nm}|a^{ND(m)})$ is
equivalent to comparing the channel causal conditioning
distributions since $P_{\theta}(y^{Nm}|a^{ND(m)}) =
P_{\theta}(y^{Nm}||x^{Nm})$ as shown below.
\begin{eqnarray}
P_{\theta}(y^{Nm}|a^{ND(m)}) & = & \prod_{i=1}^{Nm}
P_{\theta}(y_i|y^{i-1},a^{ND(m)}) \\
 & \stackrel{(a)} = & \prod_{i=1}^{Nm} P_{\theta}(y_i|y^{i-1},a^{ND(m)},z^{i-1}) \\
 & \stackrel{(b)} = & \prod_{i=1}^{Nm} P_{\theta}(y_i|y^{i-1},x^{i}) \\
 & = & P_{\theta}(y^{Nm}||x^{Nm}) \label{eqn:PYCondA=PYCausCondX}
\end{eqnarray}
In the above, $(a)$ holds since $z^{i-1}$ is a known, deterministic
function of $y^{i-1}$ and $(b)$ holds since the code-tree
$a^{ND(m)}$ together with the feedback sequence $z^{i-1}$ uniquely
determines the channel input $x^i$.

For notational convenience, the results below on the universal
decoder are stated for blocklength $n$, where $A^{D(n)}$ denotes a
code-tree of depth $n$ and $B_n$ denotes a set of such code-trees.
These results extend to the set of concatenated code-trees $B_{Nm}$
and any exceptions are described in the text. Furthermore, we
introduce the following notation: $\phi_{\theta}$ denotes the ML
decoder tuned to channel $\theta$; $P_e(\theta,\phi)$ denotes the
average (over messages and codebooks chosen uniformly from a set)
error probability when decoder $\phi$ is used over channel $\theta$;
and $P_e(\theta,\phi | {\cal C})$ denotes the average (over
messages) error probability when codebook ${\cal C}$ and decoder
$\phi$ is used over channel $\theta$.

\begin{definition}A family of channels $\{P_{Y^n||X^n,\theta}(\cdot||\cdot,\theta), \theta \in
\Theta\}$ defined over common input and output alphabets ${\cal X},
{\cal Y}$ is said to be {\it strongly separable} for the input
code-tree sets $\{B_n\}$, $B_n \subseteq {\cal X}^{(|{\cal Z}|^n
-1)/(|{\cal Z}|-1)}$, if there exists some $\mu
> 0$ that upper bounds the error exponents in the family, i.e., that
satisfies
\begin{equation}
\limsup_{n \rightarrow \infty} \sup_{\theta} - \frac{1}{n} \log
P_e(\theta, \phi_{\theta}) < \mu \label{eqn:DefStrSepMu}
\end{equation}
such that for every $\epsilon >0$ and blocklength $n$, there exists
a subexponential number $K(n)$ (that may depend on $\mu$ and on
$\epsilon$) of channels $\{\theta_{k}^{(n)}\}_{k=1}^{K(n)} \subseteq
\Theta$
\begin{equation}
\lim_{n \rightarrow \infty} \frac{1}{n} \log K(n) = 0
\label{eqn:DefStrSepK}
\end{equation}
that well approximate any $\theta \in \Theta$ in the following
sense: For any $\theta \in \Theta$ there exists $\theta_{k^*}^{(n)}
\in \Theta$, $1 \leq k^* \leq K(n)$, so that
\begin{equation}
P(y^n||x^n,\theta) \leq 2^{n \epsilon}
P(y^n||x^n,\theta_{k^*}^{(n)}), \quad \forall (x^n,y^n):
P(y^n||x^n,\theta) > 2^{-n(\mu+\log|{\cal Y}|)}
\label{eqn:DefStrSepChannels1}
\end{equation}
and
\begin{equation}
P(y^n||x^n,\theta)  \geq 2^{-n \epsilon}
P(y^n||x^n,\theta_{k^*}^{(n)}), \quad \forall (x^n,y^n):
P(y^n||x^n,\theta_{k^*}^{(n)})
> 2^{-n(\mu+\log|{\cal Y}|)} \label{eqn:DefStrSepChannels2}
\end{equation} \label{def:StrongSep}
\end{definition}

The notion of strong separability means that the family $\Theta$ is
well-approximated by a finite subset
$\{\theta_{k}^{(n)}\}_{k=1}^{K(n)} \subseteq \Theta$ of the channels
in the family. In order to prove that the family of finite-state
channels with feedback is separable, we will need a value $\mu$ that
satisfies (\ref{eqn:DefStrSepMu}). The error probability
$P_e(\theta, \phi_{\theta})$ is lower bounded by the probability
that the output sequence $Y^{Nm}$ corresponding to two different
messages is the same for a given realization of the channel and
code-tree. For a random code-tree this is lower bounded by a uniform
memoryless distribution on the channel output. Then $P_e(\theta,
\phi_{\theta}) \geq |\mathcal Y|^{-Nm}$ and a suitable candidate for
$\mu$ is $1 + \log |\mathcal Y|$. The following theorem shows the
existence of a universal decoder for a strongly separable family and
input code-tree sets $B_n$. The proof follows from the proof of
Theorem 2 in \cite{FederLapidoth98} except that we replace the
channel conditional distribution $P(y^n|x^n,\theta)$ with the causal
conditioning distribution $P(y^n||x^n,\theta)$.

\begin{theorem} \label{thrm:UnivDecoderForSeparable}
If a family of channels defined over common finite input and output
alphabets ${\cal X},{\cal Y}$ is strongly separable for the input
code-tree sets $\{B_n\}$, then there exists a sequence of rate-$R$
blocklength-$n$ codes ${\cal C}_n$ and a sequence of decoders
$\{u_n\}$ such that
\begin{equation}
\lim_{n \rightarrow \infty} \sup_{\theta} \frac{1}{n} \log \left(
\frac{P_e(\theta, u_n|{\cal C}_n)}{P_e(\theta,\phi_{\theta})}
\right) = 0
\end{equation}
\end{theorem}

The universal decoder $u_n$ in Theorem
\ref{thrm:UnivDecoderForSeparable} is given by ``merging'' the ML
decoders tuned to channels $\theta_k$, $1 \leq k \leq K(n)$, that
are used to approximate the family $\Theta$. In order to describe
the merging of the ML decoders, we first present the ranking
function $M_{\theta}$. A ML decoder tuned to the channel $\theta$
can be described by a ranking function $M_{\theta}$ defined as the
mapping
\begin{equation}
M_{\theta} : B_{Nm} \times {\cal Y}^{Nm} \rightarrow \{1,2,\hdots,
|B_{Nm}|\}
\end{equation}
where a rank of 1 denotes the code-tree $a^{ND(m)}$ that is most
likely given output $y^{Nm}$, rank 2 denotes the second most likely
code-tree, and so on. For a given received sequence $y^{Nm}$, every
code-tree in the set $B_{Nm}$ is assigned a rank. For code-trees
$a_i^{ND(m)}, a_j^{ND(m)} \in B_{Nm}$,
\begin{equation}
P_{\theta}(y^{Nm}|a_i^{ND(m)}) > P_{\theta}(y^{Nm}|a_j^{ND(m)})
\implies M_{\theta}(a_i^{ND(m)},y^{Nm}) <
M_{\theta}(a_j^{ND(m)},y^{Nm})
\end{equation}
By (\ref{eqn:PYCondA=PYCausCondX}), comparing the function
$P_{\theta}(y^{Nm}|a^{ND(m)})$ is equivalent to comparing the
channel causal conditioning distribution
$P_{\theta}(y^{Nm}||x^{Nm})$. Letting $\phi_{\theta}$ denote the ML
decoder tuned to $\theta$, we can describe the decoder as
\begin{equation}
\phi_{\theta}(y^{Nm})=w \mbox{ iff } M_{\theta}(a^{ND(m)}(w),y^{Nm})
< M_{\theta}(a^{ND(m)}(w'),y^{Nm}), \forall w' \neq w
\end{equation}
where $a^{ND(m)}(w)$ represents the code-tree chosen for message
$w$, $1 \leq w \leq e^{NR}$. In the case that multiple code-trees
maximize the likelihood $P_{\theta}(y^{Nm}|a^{ND(m)})$ for a given
$y^{Nm}$, the ranking function $M_{\theta}$ determines which
code-tree (and correspondingly message) is chosen by the decoder. In
the case that the same code-tree from $B_{Nm}$ is chosen for more
than one message, the ranks will be identical and a decoding error
will occur. Note that for a given output sequence $y^{Nm}$, the
decoder $\phi_{\theta}(y^{Nm})$ will not always return the code-tree
$a^{ND(m)} \in B_{Nm}$ for which $M_{\theta}(a^{ND(m)},y^{Nm})=1$,
since the code-tree $a^{ND(m)}$ may or may not be in the codebook.

Now consider a set of $K$ channels from the family $\Theta$, given
by $\theta_k \in \Theta, 1 \leq k \leq K$. The codebooks for these
$K$ channels will be drawn randomly from the set $B_{Nm}$. (Note
that the same set $B_{Nm}$ is used for all channels $\theta_k$
since, as shown in Lemma \ref{lemma:TypeOnCodetree}, the type
$\hat{Q}_{Nm} \in {\cal P}_N({\cal X}^{D(m)})$ is chosen independent
of the channel $P$.) The $K$ ML decoders matched to these channels,
denoted $\phi_{\theta_1}, \phi_{\theta_2}, \hdots, \phi_{\theta_K}$,
can be merged as shown in \cite{FederLapidoth98}. The merged decoder
$u_K$ is described by its ranking function $M_{u_K}$ which is a
mapping
\begin{equation}
M_{u_K} : B_{Nm} \times {\cal Y}^{Nm} \rightarrow \{1,2,\hdots,
|B_{Nm}|\}
\end{equation}
that ranks all of the code-trees in $B_{Nm}$ for each output
sequence $y^{Nm}$. The ranking $M_{u_K}$ is established for a given
$y^{Nm}$ by assigning rank 1 to the code-tree for which
$M_{\theta_1} = 1$, rank 2 to the code-tree for which
$M_{\theta_2}=1$, rank 3 to the code-tree for which
$M_{\theta_3}=1$, and so on. After considering the code-trees with
rank 1 for all $M_{\theta_k}$, the code-trees with rank 2 in
$M_{\theta_k}$, $1 \leq k \leq K$ are considered in order and added
into the ranking $M_{u_K}$. The process continues until the
code-trees with rank $|B_{Nm}|$ for all $M_{\theta_k}$ have been
assigned a rank in $M_{u_K}$. Throughout this process, if a
code-tree has already been ranked, it is simply skipped over, and
its original (higher) ranking is maintained. The rank of a code-tree
in $M_{u_K}$ can be upper bounded according to its rank in
$M_{\theta_k}$ as shown in \cite{FederLapidoth98} and stated as
follows.
\begin{equation}
M_{\theta_k}(a^{ND(m)},y^{Nm})=j \implies
M_{u_K}(a^{ND(m)},y^{Nm})\leq (j-1)K + k, \quad \forall a^{ND(m)}
\in B_{Nm}, \forall k, 1 \leq k \leq K
\end{equation}
This bound on the rank in $M_{u_K}$ implies another (looser) upper
bound.
\begin{equation} \label{eqn:UpperBoundRankMerged}
M_{u_K}(a^{ND(m)},y^{Nm})\leq K M_{\theta_k}(a^{ND(m)},y^{Nm}),
\quad \forall (a^{ND(m)}, y^{Nm}) \in B_{Nm} \times {\cal Y}^{Nm},
\forall k, 1 \leq k \leq K
\end{equation}
Equation (\ref{eqn:UpperBoundRankMerged}) can be used to upper bound
the error probability when sequences output from the channel $\theta
\in \Theta$ are decoded by the merged decoder $u_K$. This is a key
element of the proof of Theorem \ref{thrm:UnivDecoderForSeparable}.
Finally, we state the lemma below, which shows that the family of
finite-state channels defined by the causal conditioning
distribution is strongly separable. Together with Theorem
\ref{thrm:UnivDecoderForSeparable}, this establishes existence of a
universal decoder for the problem we consider, and completes our
proof of achievability.

\begin{lemma}The family of all causal-conditioning finite-state channels defined
over common finite input, output, and state alphabets ${\cal
X},{\cal Y}, {\cal S}$ is strongly separable in the sense of
Definition \ref{def:StrongSep} for any input code-tree sets
$\{B_n\}$. \label{lemma:CompoundFSCIsStronglySep}
\end{lemma}
\begin{proof}See Appendix \ref{app:CompoundFSCISStronglySep}.
\end{proof}

\section{Compound Gilbert-Elliot channel} \label{section:compoundGE}

The Gilbert-Elliot channel is a widely used example of a finite
state channel. It has a state space consisting of `good' and `bad'
states, ${\cal S}= \{G,B\}$ and in either of these two states, the
channel is a binary symmetric channel (BSC). The Gilbert-Elliot
channel is a stationary and ergodic Markovian channel, i.e.,
$P(y_i,s_i|x_i,s_{i-1},\theta)=P(s_i|s_{i-1},\theta)P(y_i|x_i,s_{i-1},\theta)$
is satisfied and the Markov process described by
$P(s_i|s_{i-1},\theta)$ is a stationary and ergodic process. For a
given channel $\theta$, the BSC crossover probability is given by
$P_B(\theta)$ for $s_i=B$ and $P_G(\theta)$ for $s_i=G$. The channel
state $S_i$ forms a stationary Markov process with transition
probabilities
\begin{eqnarray}
g(\theta) & = & P(S_i=G|S_{i{-}1}=B) =
1{-}P(S_i=B|S_{i{-}1}=B) \\
b(\theta) & = & P(S_i=B|S_{i{-}1}=G) = 1{-}P(S_i=G|S_{i{-}1}=G)
\end{eqnarray}
For a given $\theta$, the Gilbert-Elliot channel is equivalent to
the following additive noise channel
\begin{equation}
Y_i = X_i \oplus V_i
\end{equation}
where $\oplus$ denotes modulo-2 addition and $V_i \in \{0,1\}$.
Conditioned on the state process $\{S_i\}_{-\infty}^{+\infty}$, the
noise $V_i$ forms a Bernoulli process given by
\begin{equation}
\mbox{P}(V_i=1|\{S_i\}_{-\infty}^{+\infty},\theta) = \begin{cases}
P_B(\theta), & S_i=B \\ P_G(\theta), & S_i=G.
\end{cases}
\end{equation}
For a given channel $\theta$, the capacity of the Gilbert-Elliot
channel is found in \cite{Mushkin89} and is achieved by a uniform
Bernoulli input distribution.

The following example illustrates that the feedback capacity of a
channel with memory is in general {\it not} given by
\begin{equation}
C_{FB}=\inf_{\theta } C_\theta,
\end{equation}
as in the memoryless case.

\begin{example}\cite{LapidothTelatar98} \label{example:compoundGE}
Consider the example of a Gilbert-Elliot channel where
$P_G(\theta)=0,P_B(\theta)=0.5, b(\theta)=g(\theta)=2^{-\theta}$ for
$\theta= 1,2,3....$ with feedback. The compound feedback capacity of
this channel is zero because assuming that we start in the bad
state, for any blocklength $n$, the channel that corresponds to
$\theta=n$, will remain in the bad state for the duration of the
transmission with probability $(1-2^{-n})^n>1-n2^{-n}\geq
\frac{1}{2}$. While the channel is in the bad state the probability
of error for decoding the message is positive with or without
feedback, hence no reliable communication is possible.

However if we fix $\theta$, then the capacity $C_{\theta}$ is at
least $1-h_b(\frac{1}{4})$, because we can use a deep enough
interleaver to make the channel look like memoryless BSC with
crossover probability $\frac{1}{4}$.
\end{example}

A Gilbert-Elliot channel is described by the four parameters
$g(\theta),b(\theta),P_G(\theta),$ and $P_B(\theta)$ that lie
between 0 and 1 and for any fixed $n$, $P(y^n||x^n,s_0)$ is
continuous in those parameters. The continuity of $P(y^n||x^n,s_0)$
follows from the fact that $P(y_i,s_i|x_{i},s_{i-1})$ is continuous
in the four parameters for any $i\geq 1$, and also because (as shown
in Appendix \ref{app:CompoundFSCISStronglySep} in Eqns.
(\ref{eqn:EquivFederLapidoth63}) and
(\ref{eqn:EquivFederLapidoth64})) we can express $P(y^n||x^n,s_0)$
as
\begin{eqnarray}
P(y^n||x^n,s_0)&=&\sum_{s^n} P(y^n,s^n||x^n,s_0) \nonumber \\
&=&\sum_{s^n} \prod_{i=1}^n P(y_i,s_i|x_i,s_{i-1}).
\end{eqnarray}

Let us denote by $\overline {\Theta}$ the closure of the family of
channels. Hence instead of $\inf_{\theta\in \Theta}$ we can write
$\min_ {\theta\in {\overline\Theta}}$ since $\overline {\Theta}$ is
compact and since $\mathcal I(Q;P)$ is continuous in $P$. Now, let
$Q_u(x^n)$ denote the uniform distribution over $\mathcal X^n$. We
have
\begin{eqnarray}
\max_{Q}\min_{s_0,\theta} \mathcal I(Q;P)&\stackrel{(a)}{\leq}&
\min_{s_0,\overline{\theta}} \max_{Q}
\mathcal I(Q;P)\nonumber \\
&\stackrel{(b)}{=}& \min_{s_0,\overline{ \theta}}
I(Q_u;P)\nonumber \\
\end{eqnarray}
where $(a)$ follows from the fact that $\max \min\leq \min \max $
and $(b)$ follows from the fact that for any channel a uniform
distribution maximizes its capacity. Therefore we can restrict the
maximization to the uniform distribution $Q_u$ instead of
$Q(x^n||y^{n-1})$. Hence feedback does not increase the capacity of
the compound Gilbert-Elliot channel. This result holds for any
family of FSCs for which the uniform input distribution achieves the
capacity of each channel in the family and is closely related to
Alajaji's result \cite{Alajaji95} that feedback does not increase
the capacity of discrete additive noise channels.

%
%

\section{Feedback capacity is positive if and only if capacity without feedback is
positive} \label{section:IFF}

In this section we show that the capacity of a compound channel that
consists of stationary and uniformly ergodic Markovian channels is
positive if and only if it is positive for the case that feedback is
allowed. The intuition of this result comes mainly from Lemma
\ref{lemma:directed0iff} that states that
\begin{equation}
\max_{Q_{X^n||Y^{n-1}}} I(X^n \to Y^n) =0 \iff \max_{Q_{X^n}} I(X^n
\to Y^n) =0.
\end{equation}
The reason our proof is restricted to the family of channels that
are stationary and uniformly ergodic Markovian is because for this
family of channels we can show that the capacity is zero only if for
every finite $n$,
\begin{equation}
\max_{Q_{X^n||Y^{n-1}}} \inf_{\theta} I(X^n \to Y^n | \theta) = 0.
\end{equation}

A stationary and ergodic Markovian channel is a FSC where the state
of the channel is a stationary and ergodic Markov process that is
not influenced by the channel input and output. In other words, the
conditional probability of the channel output and state given the
input and previous state is given by
\begin{equation}
P(y_i,s_i|x_i,s_{i-1},\theta)=P(s_i|s_{i-1},\theta)P(y_i|x_i,s_{i-1},\theta)
\label{eqn:MarkovChannel}
\end{equation}
where the Markov process, described by the transition probability
$P(s_i|s_{i-1},\theta)$, is stationary and ergodic. We say that the
family of channels is {\it uniformly ergodic} if all channels in the
family are ergodic and for all $\epsilon>0$ there exists an
$M(\epsilon)$ such that for all $n>M$
\begin{equation}
|\Pr(S_n=s|s_0,\theta)-P(s|\theta)|\leq \epsilon, \; \; \forall
s_0\in \mathcal S,s\in \mathcal S, \theta\in \Theta
\end{equation}
where $P(s|\theta)$ is the stationary (equilibrium) distribution of
the state for channel $\theta$. We define the sequence
$C_n^{Markovian}$ as
\begin{equation}\label{eqn:C_nMarkov_def}
C_n^{Markovian}= \max_{Q_{X^n||Z^{n-1}}} \inf_{ \theta} \frac{1}{n}
I(X^n\to Y^n| \theta).
\end{equation}

\begin{theorem} \label{thrm:ZeroFBCIffZeroNFBC}
The channel capacity of a family of stationary and uniformly ergodic
Markovian channels is positive if and only if the feedback capacity
of the same family is positive.
\end{theorem}

Since a memoryless channel is a FSC with only one state, the theorem
implies that the feedback capacity of a memoryless compound channel
is positive if and only if it is positive without feedback. The
theorem also implies that for a stationary and ergodic
point-to-point channel (not compound), feedback does not increase
the capacity for cases that the capacity without feedback is zero.
The stationarity of the channels in Theorem
\ref{thrm:ZeroFBCIffZeroNFBC} is not necessary since according to
our achievability definition, if a rate is less than the capacity,
it is achievable regardless of the initial state. We assume
stationarity here in order to simplify the proofs. The uniform
ergodicity is essential to the proof that is provided here but there
are also other family of channels that have this property. For
instance, for the regular point-to-point Gaussian channel this
result can be concluded from factor two result that claims that
feedback at most doubles capacity (c.f.,
\cite{Pinsker_feedback_double,Ebert,Cover89}). The proof of Theorem
\ref{thrm:ZeroFBCIffZeroNFBC} is based on the following lemmas. We
refer the reader to Appendix \ref{app:LemmasForIFF} for the proofs
of these lemmas.

\begin{lemma}\label{lemma:QkQm_inequality}
For any channel with feedback, if the input to the channel is
distributed according to
\begin{equation*} Q(x^n||z^{n-1}) =
Q(x_1^k||z_1^{k-1})Q(x_{k+1}^{n}||z_{k+1}^{n-1}),
\end{equation*}
then
\begin{equation}\label{e_in_sup} I(X^n\to Y^n) \geq
I(X^k\to Y^k) + I(X_{k+1}^n\to Y_{k+1}^n).
\end{equation}
\end{lemma}

\begin{lemma}\label{lemma:capacity_stationary_ergodic_Markovian}
The feedback capacity of a family of stationary and uniformly
ergodic Markovian channels is
\begin{equation}
\lim_{n \rightarrow \infty} C_n^{Markovian}.
\end{equation}
The limit of $C_n^{Markovian}$ exists and is equal to $\sup_{n}
C_n^{Markovian}$.
\end{lemma}

\begin{lemma}\label{lemma:directed0iff}
Let the input distribution to an arbitrary channel be uniform over
the input $\mathcal X^n$, i.e., $Q(x^n)=\frac{1}{|\mathcal X|^n}$.
If under this input distribution $I(X^n\to Y^n) {=}0$, then the
channel has the property that $P(y^n||x^n)=P(y^n)$ for all $x^n\in
\mathcal X^n,y^n\in \mathcal Y^n$ and this implies that
\begin{equation}\label{eqn:0iff}
\max_{Q_{X^n||Y^{n-1}}} I(X^n\to Y^n) =0.
\end{equation}
\end{lemma}

{\it Proof of Theorem \ref{thrm:ZeroFBCIffZeroNFBC}}: Let $C_{NFB}$
denote the capacity without feedback and $C_{FB}$ denote the
capacity with feedback. $C_{NFB}=0 \Leftarrow C_{FB}=0$ is trivial.
To show that $C_{NFB}=0 \implies C_{FB}=0$, we use Lemma
\ref{lemma:capacity_stationary_ergodic_Markovian} to conclude that
since $C_{NFB}=0$ then $\sup_{n} C_n^{Markovian}=0$ and therefore
for any $n\geq1$,
\begin{equation}\label{eqn:Compound_no_feedback2}
\max_{Q_{X^n}} \inf_{\theta}  I(X^n\to Y^n|\theta)=0.
\end{equation}
In order to conclude the proof, we show that if
(\ref{eqn:Compound_no_feedback2}) holds, then it also holds when we
replace $Q_{X^n}$ by $Q_{X^n||Y^{n-1}}$. Since $ I(X^n\to Y^n)$ is
continuous in $P(y^n||x^n)$ and since the set $\Theta$ is a subset
of the unit simplex which is bounded, then the infimum over the set
$\Theta$ can be replaced by the minimum over the closure of the set
$\Theta$. Since (\ref{eqn:Compound_no_feedback2}) holds also for the
case that $Q_{X^n}$ is restricted to be the uniform distribution,
then Lemma \ref{lemma:directed0iff} implies that the channel that
satisfies $P(y^n||x^n)=P(y^n)$ for all $x^n\in \mathcal X^n,y^n\in
\mathcal Y^n$ is in the closure of $\Theta$ and therefore
\begin{equation}
\max_{Q_{X^n||Y^{n-1}}} \inf_{\theta} I(X^n\to Y^n |\theta)=0.
\end{equation}
\hfill \QED

\section{Feedback capacity of the memoryless compound channel}
\label{section:memorylesscompound}

Recall that the capacity of the memoryless compound channel (without
feedback) is \cite{Blackwell59Compound,Wolfowitz59}
\begin{equation}
\max_{Q_X}\inf_{\theta} \mathcal I(Q_X;P_{Y|X,\theta}).
\end{equation}
Wolfowitz also showed \cite{Wolfowitz64} that when $\theta$ is known
to the encoder, the capacity of the memoryless compound channel is
given by switching the $\inf$ and the $\max$, i.e.,
\begin{equation}
\inf_{\theta} \max_{Q_X} \mathcal I(Q_X;P_{Y|X,\theta}).
\label{eqn:CapacityEncoderKnowsTheta}
\end{equation}
In this section we make use of Theorem \ref{thrm:FBCapCompoundFSC}
to show that (\ref{eqn:CapacityEncoderKnowsTheta}) is equal to the
feedback capacity of the memoryless compound channel.

\subsection{Finite family of memoryless channels}
Based on Wolfowitz's result it is straightforward to show that if
the family of memoryless channels is finite, $| \Theta| < \infty$,
then the feedback capacity of the compound channel is given by
switching the $\max$ and the $\min$,
\begin{equation}
\min_{\theta} \max_{Q_X} \mathcal I(Q_X;P_{Y|X,\theta}).
\end{equation}
This result can be achieved in two steps. Given a probability of
error $P_e>0$, first, the encoder will use $M$ uses of the channels
in order to estimate the channel with probability of error less than
$\frac{P_e}{2}$. Since the number of channels is finite such an $M$
exists. In the second step the encoder will use a coding scheme with
blocklength $N$ adapted for the estimated channel to obtain an error
probability that is smaller than $\frac{P_e}{2}$. Hence we get that
the total error of the code of length $M+N$ is smaller than $P_e$.

\subsection{Arbitrary family of memoryless channels}
For the case that the number of channels is infinite, the argument
above does not hold, since there is no guarantee that for any
$P_e>0$ there exists a blocklength $n(P_e)$ such that a $(e^{nR},n)$
code achieves an error less than $P_e$ for all channels in the
family. \footnote{In a private communication with A. Tchamkerten
\cite{Tchamkerten07}, it was suggested that the feedback capacity of
the memoryless compound channel with an infinite family can also be
established using the results in \cite{FederLapidoth98} (which show
that the family of all discrete memoryless channels is strongly
separable). The family is finitely quantized, a training scheme is
used to estimate the appropriate quantization cell, the coding is
performed according to the representative channel of that cell and
the decoding is done universally as in \cite{FederLapidoth98}.}
However, we are able to establish the feedback capacity using our
capacity theorem for the compound FSC, and the result is stated in
the following theorem.
\begin{theorem} \label{thrm:CapacityMemorylessCompound}
The feedback capacity of the memoryless compound channel is
\begin{equation}
\inf_{\theta} \max_{Q_X} \mathcal I(Q_X;P_{Y|X,\theta}).
\end{equation}
\end{theorem}
Theorem \ref{thrm:CapacityMemorylessCompound} is a direct result of
Theorem \ref{thrm:FBCapCompoundFSC} and the following lemma.
\begin{lemma}\label{l_memoryless_eq}
For a family $\Theta$ of memoryless channels we have
\begin{equation}
\lim_{n\to \infty} \frac{1}{n}\max_{Q_{X^n||Y^{n-1}}} \inf_{\theta}
\mathcal I(Q_{X^n||Y^{n-1}}; P_{Y^n||X^n,\theta})=\inf_{\theta}
\max_{Q_X} \mathcal I(Q_X;P_{Y|X,\theta})
\end{equation}
\end{lemma}
The proof of Lemma \ref{l_memoryless_eq} requires two lemmas, which
we state below. The proofs of Lemmas \ref{l_DoNotLooseMuch} and
\ref{l_estimate_channel} are found in Appendix
\ref{app:LemmasForMemorylessCompound}.

\begin{lemma}\label{l_DoNotLooseMuch}
Let $Q_X^{1}=\arg\max_{Q_X}\mathcal I(Q_X,P_{Y|X,\theta_1})$ and
$Q_X^{2}=\arg\max_{Q_X}\mathcal I(Q_X,P_{Y|X,\theta_2})$. For two
conditional distributions $P_{Y|X,\theta_1}$ and $P_{Y|X,\theta_2}$
with
\begin{equation}
\Delta=||P_{Y|X,\theta_1}-P_{Y|X,\theta_2}||_1=\sum_{x\in \mathcal
X, y\in \mathcal Y}
|P_{Y|X,\theta_1}(y|x,\theta_1)-P_{Y|X,\theta_2}(y|x,\theta_2)|
\end{equation}
there exists an upper bound
\begin{equation}\label{e_l_diff}
|\mathcal I(Q_X^{2},P_{Y|X,\theta_1})-\mathcal
I(Q_X^{1},P_{Y|X,\theta_1})|\leq \eta(\Delta)
\end{equation}
where  $\eta(\Delta)\to 0$ as $\Delta\to 0$.
\end{lemma}


\begin{lemma}\label{l_estimate_channel}
For any $\delta>0$, any $\epsilon>0$ and any channel $P_{Y|X}$,
there exists an $M$ such that we can choose a channel
$P_{Y|X,\hat\theta}$ as a function of $M$ inputs and outputs such
that
\begin{equation}\label{e_Pr_Delta}
\Pr\{\Delta>\epsilon\}\leq \delta,
\end{equation}
where $\Delta$ denotes the $L_1$ distance between the estimated
channel $ P_{Y|X,\hat \theta}$ and the actual channel $P_{Y|X}$,
i.e.,
\begin{equation}
\Delta=\sum_{x\in \mathcal X, y\in \mathcal Y} |
P_{Y|X,\hat\theta}(y|x,\hat \theta)-P_{Y|X}(y|x)|.
\end{equation}
\end{lemma}

{\it Proof of Lemma \ref{l_memoryless_eq}:} We prove the equality by
showing the following two inequalities hold:
\begin{eqnarray}
\frac{1}{n}\max_{Q_{X^n||Y^{n-1}}} \inf_{\theta} \mathcal
I(Q_{X^n||Y^{n-1}};P_{Y^n||X^n,\theta}) &\leq& \inf_{\theta}
\max_{Q_X} \mathcal I(Q_X;P_{Y|X,\theta}),\label{e_memoryless1} \\
\frac{1}{n}\max_{Q_{X^n||Y^{n-1}}} \inf_{\theta} \mathcal
I(Q_{X^n||Y^{n-1}};P_{Y^n||X^n,\theta}) &\geq &\inf_{\theta}
\max_{Q_X} \mathcal I(Q_X;P_{Y|X,\theta})-\epsilon_n,
\label{e_memoryless2}
\end{eqnarray}
where $\epsilon_n\to 0$ as $n\to \infty$. Inequality
(\ref{e_memoryless1}) is proved by the fact that $\max \inf$ is less
than or equal to  $\inf \max$ and by the fact that for a memoryless
channel an i.i.d input maximizes the directed information.
\begin{eqnarray}
\lefteqn{\frac{1}{n}\max_{Q_{X^n||Y^{n-1}}} \inf_{\theta} \mathcal
I(Q_{X^n||Y^{n-1}};P_{Y^n||X^n,\theta})}\nonumber \\ &\leq&
\frac{1}{n}\inf_{\theta} \max_{Q_{X^n||Y^{n-1}}}  \mathcal
I(Q_{X^n||Y^{n-1}};P_{Y^n||X^n,\theta})\nonumber \\
&=& \inf_{\theta} \max_{Q_X}  \mathcal I(Q_X;P_{Y|X,\theta})
\end{eqnarray}

In order to prove inequality (\ref{e_memoryless2}) we consider the
following input distribution. The first $M$ inputs are used to
estimate the channel and we denote the estimated channel as  $\hat
\theta$. After the first $M$ inputs, the input distribution is the
i.i.d distribution that maximizes the mutual information between the
input and the output for the channel $\hat \theta$. According to
Lemma \ref{l_estimate_channel}, we can estimate the channel to
within an $L_1$ distance smaller than $\epsilon>0$ with probability
greater than $1-\delta$, where $\delta>0$. According to Lemma
\ref{l_DoNotLooseMuch}, by adjusting the input distribution to a
channel that is at $L_1$ distance less than $\epsilon$ from the
actual channel in use, we lose an amount that goes to zero as
$\epsilon\to 0$. Under the input distribution described above we
have the following sequence of inequalities.
\begin{eqnarray}
\lefteqn{\frac{1}{n}\max_{Q_{X^n||Y^{n-1}}} \inf_{\theta} \mathcal
I(Q_{X^n||Y^{n-1}};P_{Y^n||X^n,\theta})}\nonumber
\\
&\stackrel{(a)}{=}&{\frac{1}{n}\max_{Q_{X^n||Y^{n-1}}} \inf_{\theta}
I(X^n\to Y^n|\theta )}\nonumber \\
&\stackrel{(b)}{\geq}&\frac{1}{n}\max_{Q_{X^n||Y^{n-1}}}
\inf_{\theta}
\sum_{i=M(\delta,\epsilon)+1}^n I(X^i;Y_i|Y^{i-1})\nonumber \\
&\stackrel{(c)}{\geq}&\frac{1}{n}\max_{Q_{X^n||Y^{n-1}}}
\inf_{\theta}
\sum_{i=M+1}^n I(X_{M+1}^i;Y_i|Y^{i-1}, X^M)\nonumber \\
&\stackrel{(d)}{=}&\frac{1}{n}\max_{Q_{X^n||Y^{n-1}}} \inf_{\theta}
\sum_{i=M+1}^n I(X_{M+1}^i;Y_i|Y_{M+1}^{i-1}, X^M,Y^M,\hat
\Theta(X^M,Y^M))\nonumber \\
&\stackrel{(e)}{\geq}&\frac{1}{n}\max_{Q_{X|\hat \theta}}
\inf_{\theta} (n-M) I(X;Y|\theta, \hat
\Theta )\nonumber \\
&\stackrel{(f)}{=}&\frac{1}{n}\max_{Q_{X|\hat \theta}}
\inf_{\theta} (n-M) \sum_{\hat \theta_{\epsilon}}P(\hat \theta)\mathcal I(Q_{X|\hat \theta};P_{Y|X,\theta} )\nonumber \\
&\stackrel{(g)}{\geq}&\frac{1}{n}\max_{Q_{X|\theta}} \inf_{\theta}
(n-M)
(1-\delta)\mathcal I(Q_{X|\theta};P_{Y|X,\theta})-\eta(\epsilon))\nonumber \\
&\stackrel{(h)}{=}&\frac{1}{n} \inf_{\theta} \max_{Q_X}(n-M)
(1-\delta)\mathcal I(Q_X;P_{Y|X,\theta})-\eta(\epsilon))
\end{eqnarray}

\begin{itemize}
\item [(a)] and (f) follow from a change of notation.
\item [(b)] follows the fact that we sum fewer elements. The
parameter $M$ is a function of $\epsilon>0$ and $\delta>0$ and is
determined according to Lemma \ref{l_estimate_channel}. For brevity
of notation we denote $M(\epsilon,\delta)$ simply as $M$.
\item [(c)] follows from the fact that $H(Y_i|Y^{i-1}) \geq
H(Y_i|Y^{i-1},X^M)$.
\item[(d)] follows from the fact that the estimated channel is a random variable denoted as $\hat
\Theta$ and it is a deterministic function of $X^M,Y^M$ as described
in Lemma \ref{l_estimate_channel}.
\item [(e)] follows by restricting the input distribution $Q_{X^n||Y^{n-1}}$ to one that uses
first $M$ uses of the channel to estimate as described in Lemma
\ref{l_estimate_channel}, and then uses an i.i.d distribution, i.e.,
for $i>M$, $Q(x_i|x^{i-1},y^{i-1})=Q(x_i|x^{i-1},y^{i-1},\hat
\theta(x^M,y^M)))= Q(x_i|\hat \theta)$.
\item[(g)] follows from the fact that with probability $1-\delta$
we have that the $L_1$ distance $||P_{Y|X,\theta}-P_{Y|X,\hat
\theta}||_1\leq \epsilon$ and by applying Lemma
\ref{l_DoNotLooseMuch}, which states that for this case we lose
$\eta(\epsilon)$ where $\eta(\epsilon)\to 0$ as $\epsilon\to 0$ .
\item[(h)] follows from the fact that $\inf_{\theta}
\max_{Q_X}$ is identical to $\max_{Q_{X|\theta}} \inf_{\theta}$.
\end{itemize}
Finally, since $M$ is fixed for any $\epsilon>0$, $\delta>0$ then we
can achieve any value below $\inf_{\theta} \max_{Q_X}\mathcal
I(Q_X;P_{Y|X,\theta})$ for large $n$. Therefore inequality
(\ref{e_memoryless2}) holds. \hfill \QED

\section{Conclusion}

The compound channel is a simple model for communication under
channel uncertainty. The original work on the memoryless compound
channel without feedback characterizes the capacity
\cite{Blackwell59Compound,Wolfowitz59}, which is less than the
capacity of each channel in the family, but the reliability function
remains unknown. An adaptive approach to using feedback on an
unknown memoryless channel is proposed in
\cite{TchamkertenTelatar06}, where coding schemes that universally
achieve the reliability function (the Burnashev error exponent) for
certain families of channels (e.g., for a family of binary symmetric
channels) are provided. By using the variable-length coding approach
in \cite{TchamkertenTelatar06}, the capacity of the channel in use
can be achieved. In our work, we consider the use of fixed length
block codes and aim to ensure reliability for every channel in the
family; as a result, our capacity is limited by the infimum of the
capacities of the channels in the family. For the compound channel
with memory that we consider, we have characterized an achievable
random coding exponent, but the reliability function remains
unknown.

The encoding and decoding schemes used in proving our results have a
number of practical limitations, including the memory requirements
for storing codebooks consisting of concatenated code-trees at both
the transmitter and receiver as well as the complexity involved in
merging the maximum-likelihood decoders tuned to a number of
channels that is polynomial in the blocklength. As such, our work
motivates a search for more practical schemes for feedback
communication over the compound channel with memory.

\appendices

\section{Proof of Proposition \ref{prop:existenceC}} \label{app:ExistenceProof}

The proposition is nearly identical to \cite[Proposition
1]{LapidothTelatar98} except that we replace $I(X^n;Y^n|s_0,\theta)$
by $I(X^n \to Y^n|s_0,\theta)$ and $Q(x^n)$ by $Q(x^n||z^{n-1})$
using results from \cite{Permuter06_feedback_submit} on directed
mutual information and causal conditioning. We first prove the
following lemma, which is needed in the proof of Proposition
\ref{prop:existenceC}. The lemma shows that directed information is
uniformly continuous in $Q_{X^n||Y^{n-1}}$. For our time-invariant
deterministic feedback model, $Q(x^n||y^{n-1})=Q(x^n||z^{n-1})$, and
the lemma holds for any such feedback.

\begin{lemma} \label{lemma:continuityDirectedInfo}
{\it (Uniform continuity of directed information)} If
$Q^1_{X^n||Y^{n-1}}$ and $Q^2_{X^n||Y^{n-1}}$ are two causal
conditioning distributions such that
\begin{equation}
\sum_{x^n\in {\cal X}^n,y^n \in {\cal
Y}^n}|Q^1(x^n||y^{n-1})-Q^2(x^n||y^{n-1})|\leq \Delta \leq
\frac{1}{2}
\end{equation}
then for a fixed $P_{Y^n||X^n}$
\begin{equation}
|\mathcal I(Q^1_{X^n||Y^{n-1}}; P_{Y^n||X^n}) - \mathcal
I(Q^2_{X^n||Y^{n-1}}; P_{Y^n||X^n})|\leq-\Delta \log
\frac{\Delta}{|{\cal Y}^n|^2}.
\end{equation}
\end{lemma}

\begin{proof}
Directed information can be expressed as a difference between two
terms $I(X^n\to Y^n)=H(Y^n)-H(Y^n||X^n)$. Let us consider the total
variation of $P^1_{Y^n}(\cdot)-P^2_{Y^n}(\cdot)$,
\begin{eqnarray}
\sum_{y^n}|P^1(y^n)-P^2(y^n)|&=&\sum_{y^n}\left|\sum_{x^n} P^1(x^n,y^n) -P^2(x^n,y^n)\right|\nonumber \\
&=&\sum_{y^n}\left|\sum_{x^n} Q^1(x^n||y^{n-1})P(y^n||x^n) -Q^2(x^n||y^{n-1})P(y^n||x^n)\right|\nonumber \\
&\leq&\sum_{y^n}\sum_{x^n} P(y^n||x^n)\left|Q^1(x^n||y^{n-1}) -Q^2(x^n||y^{n-1})\right|\nonumber \\
&\leq&\sum_{y^n}\sum_{x^n} \left|Q^1(x^n||y^{n-1}) -Q^2(x^n||y^{n-1})\right|\nonumber \\
&\leq& \Delta
\end{eqnarray}
By invoking the continuity lemma of entropy  \cite[Theorem 2.7,
p33]{Csiszar81} we get,
\begin{equation}\label{eq_Hdiff1}
|H^1(Y^n)-H^2(Y^n)|\leq -\Delta \log \frac{\Delta}{|{\cal Y}^n|}
\end{equation}
where $H^1(Y^n)$ and $H^2(Y^n)$ are the entropies induced by
$P^1_{Y^n}(\cdot)$ and $P^2_{Y^n}(\cdot)$, respectively. Now let us
consider the difference $H^1(Y^n||X^n)-H^2(Y^n||X^n)$.
\begin{eqnarray}\label{eq_Hdiff2}
\lefteqn{| H^1(Y^n||X^n)-H^2(Y^n||X^n)|}\nonumber \\& = &\left
|\sum_{x^n,y^n}-P^1(x^n,y^n)\log P(y^n||x^n)+P^2(x^n,y^n)\log
P(y^n||x^n)\right| \nonumber \\
& = &\left|\sum_{x^n,y^n} -P(y^n||x^n)Q^1(x^n||y^{n-1})\log
P(y^n||x^n)+P(y^n||x^n)Q^2(x^n||y^{n-1})\log
P(y^n||x^n)\right|\nonumber \\
& = &\left|\sum_{x^n,y^n} -P(y^n||x^n)\log
P(y^n||x^n)\left(Q^1(x^n||y^{n-1})-Q^2(x^n||y^{n-1})\right)\right|\nonumber \\
& \leq &\left|\sum_{x^n,y^n} -P(y^n||x^n)\log
P(y^n||x^n)\left|Q^1(x^n||y^{n-1})-Q^2(x^n||y^{n-1})\right| \right| \nonumber \\
& \leq &\left(\sum_{x^n,y^n} -P(y^n||x^n)\log
P(y^n||x^n)\right) \left( \sum_{x^n,y^n}|Q^1(x^n||y^{n-1})-Q^2(x^n||y^{n-1})|\right)\nonumber \\
&\leq&\log {|{\cal Y}^n|} \Delta
\end{eqnarray}
By combining inequalities (\ref{eq_Hdiff1}) and (\ref{eq_Hdiff2}) we
conclude the proof of the lemma.
\end{proof}

By Lemma \ref{lemma:continuityDirectedInfo}, $I(X^n \to Y^n |
s_0,\theta)$ is uniformly continuous in $Q_{X^n||Z^{n-1}}$. Since
$Q_{X^n||Z^{n-1}}$ is a member of a compact set, the maximum over
$Q_{X^n||Z^{n-1}}$ is attained and $C_n$ is well-defined.

Next, we invoke a result similar to \cite[Lemma
5]{LapidothTelatar98}. Given integers $k$ and $m$ such that $k+m=n$,
input sequences $x_1^k = (x_1,\hdots,x_k)$ and
$x_{k+1}^{n}=(x_{k+1},\hdots,x_n)$ with corresponding output
sequences $y_1^k$ and $y_{k+1}^{n}$, let $Q_{X^n||Z^{n-1}}$ be
defined as
\begin{equation*}
Q(x^n||z^{n-1}) = Q(x_1^k||z_1^{k-1})Q(x_{k+1}^{n}||z_{k+1}^{n-1}).
\end{equation*}
Then
\begin{equation*}
\inf_{s_0,\theta} I(X^n \to Y^n | s_0,\theta) \geq \inf_{s_0,\theta}
I(X_1^k \to Y_1^k | s_0,\theta) + \inf_{s_0,\theta} I(X_{k+1}^{n}
\to Y_{k+1}^{n} | s_k,\theta) - \log|{\cal S}|.
\end{equation*}
This result follows from \cite[Lemma 5]{LapidothTelatar98} and
\cite[Lemma 5]{Permuter06_feedback_submit}.

Finally, if we let $Q(x_1^k||z_1^{k-1})$ and
$Q(x_{k+1}^{n}||z_{k+1}^{n-1})$ achieve the maximizations in $C_k$
and $C_m$, respectively, then we have
\begin{eqnarray*}
n C_n & \geq & \inf_{s_0,\theta} I(X^n \to Y^n | s_0,\theta)
\\
 & \geq & \inf_{s_0,\theta}
I(X_1^k \to Y_1^k | s_0,\theta) + \inf_{s_0,\theta} I(X_{k+1}^{n}
\to
Y_{k+1}^{n} | s_k,\theta) - \log|{\cal S}| \\
 & = & kC_k + mC_m - \log|{\cal S}|,
\end{eqnarray*}
or equivalently,
\begin{equation*}
n \hat{C}_n \geq k \hat{C}_k + m \hat{C}_m.
\end{equation*}
Clearly $\lim_{n \rightarrow \infty} C_n = \lim_{n \rightarrow
\infty} \hat{C}_n$, and by the convergence of a super-additive
sequence, $\lim_{n \rightarrow \infty} \hat{C}_n = \sup_{n}
\hat{C}_n$.

\section{Proof of Theorem \ref{thrm:AchievabilityThetaKnown}}
\label{app:AchievabilityThetaKnown}

The theorem is proved through a collection of results in
\cite{LapidothTelatar98} and \cite{Permuter06_feedback_submit}. Let
$P_{e,w}^{n}(\theta)$ denote the error probability of the ML decoder
when a random code-tree of blocklength $n$ is used at the encoder.
\begin{equation}
P_{e,w}^{n}(\theta) = \sum_{y^n \in {\cal Y}^n : \hat{w} \neq
w}P(y^n||x^n(w,z^{n-1}),\theta)
\end{equation}
The following corollary to \cite[Theorem
8]{Permuter06_feedback_submit} bounds the expected value
$E[P_{e,w}^{n}(\theta)]$, where the expectation is with respect to
the randomness in the code. The result holds for any initial state
$s_0$.
\begin{corollary} \label{cor:boundEProbError}
Suppose that an arbitrary message $w$, $1 \leq w \leq e^{nR}$,
enters the encoder with feedback and that ML decoding tuned to
$\theta$ is employed. Then the average probability of decoding error
over the ensemble of codes is bounded, for any choice of $\rho$, $0
< \rho \leq 1$, by
\begin{equation}
E[P_{e,w}^{n}(\theta)] \leq (e^{nR}-1)^{\rho} \sum_{y^n} \left[
\sum_{x^n} Q(x^n||z^{n-1}) P(y^n||x^n,\theta)^{\frac{1}{1+\rho}}
\right]^{1+\rho}.
\end{equation}
\end{corollary}
\begin{proof} Identical to \cite[Proof of Theorem
8]{Permuter06_feedback_submit} except that $P(y^n||x^n)$ is replaced
by $P(y^n||x^n,\theta)$. \end{proof}

Next, we let $P_{e}^{n}(s_0,\theta)$ denote the average (over
messages) error probability incurred when a code-tree of blocklength
$n$ is used over channel $\theta$ with initial state $s_0$. Using
Corollary \ref{cor:boundEProbError}, we can bound
$P_{e}^{n}(s_0,\theta)$ as in the following Corollary to
\cite[Theorem 9]{Permuter06_feedback_submit}

\begin{corollary} \label{cor:boundProbErrorFE}
For a compound FSC with $|{\cal S}|$ states where the codewords are
drawn independently according to a given distribution $Q_n \in {\cal
P}({\cal X}^n || {\cal Z}^{n-1})$ and ML decoding tuned to $\theta$
is employed, the average probability of error
$P_{e}^{n}(s_0,\theta)$ for any initial state $s_0 \in {\cal S}$,
channel $\theta \in \Theta$, and $\rho$, $0 \leq \rho \leq 1$ is
bounded as
\begin{equation}
P_{e}^{n}(s_0,\theta) \leq |{\cal S}|
\exp\left(-n(F^n(\rho,Q_n,\theta) -\rho R ) \right)
\end{equation}
where
\begin{multline}
F^{n}(\rho,Q_n, \theta) = \frac{-\rho \log|{\cal S}|}{n}
+ \min_{s_0} E_0(\rho,Q_n, s_0,\theta) \\
E_0(\rho,Q_n, s_0,\theta) = - \frac{1}{n} \log \sum_{y^n} \left[
\sum_{x^n} Q_n P(y^n||x^n,s_0,\theta)^{\frac{1}{1+\rho}}
\right]^{1+\rho}
\end{multline}
\end{corollary}
\begin{proof}
Identical to \cite[Proof of Theorem 9]{Permuter06_feedback_submit}
except for: (i) we replace $P(y^n||x^n,s_0)$ by
$P(y^n||x^n,s_0,\theta)$, (ii) we consider the error averaged over
all messages (rather than the error for an arbitrary message $w$),
and (iii) we assume a fixed input distribution $Q_{X^n||Z^{n-1}}$
rather than minimizing the error probability over all
$Q_{X^n||Z^{n-1}}$.
\end{proof}

The two results stated above provide us with a bound on the error
probability, however, the bound depends on the channel $\theta$ in
use. Instead, we would like to bound the error probability uniformly
over the class $\Theta$. To do so we cite the following two lemmas
from previous work.

\begin{lemma} \label{lemma:Fsuperadditive}
Given $Q_k \in {\cal P}({\cal X}^k||{\cal Z}^{k-1})$ and $Q_m \in
{\cal P}({\cal X}^m||{\cal Z}^{m-1})$, let $m=n-k$ and define
\begin{equation}
Q_n(x_1^n||z_1^{n-1}) = Q_k(x_1^k||z_1^{k-1})
Q_m(x_{k+1}^{n}||z_{k+1}^{n-1}).
\end{equation}
Then $F^{n}(\rho,Q_n,\theta)$ as defined in Corollary
\ref{cor:boundProbErrorFE} satisfies
\begin{equation}
F^{n}(\rho,Q_n, \theta) \geq \frac{k}{n}F^{k}(\rho,Q_k, \theta) +
\frac{m}{n} F^{m}(\rho,Q_m, \theta).
\end{equation}
\end{lemma}
\begin{proof}
Identical to \cite[Proof of Lemma 11]{Permuter06_feedback_submit}
except that we replace $P(y^n||x^n,s_0)$ by
$P(y^n||x^n,s_0,\theta)$.
\end{proof}

\begin{lemma} \label{lemma:boundE}
\begin{equation}
E_0(\rho,Q_n, s_0,\theta) \geq \frac{1}{n} \rho \mathcal
I(Q_n;P_{Y^n||X^n,s_0,\theta}) - \frac{1}{2n} \rho^2 \left(
\log(e|{\cal Y}|) \right)^2
\end{equation}
\end{lemma}
\begin{proof}
The lemma follows from \cite[Lemma 2]{LapidothTelatar98}, which
holds for a channel $P$ and input distribution $Q$ satisfying
$\sum_{x^n}Q(x^n||z^{n-1})=1$ and
$\sum_{x^n,y^n}Q(x^n||z^{n-1})P(y^n||x^n)=1$.
\end{proof}

We now follow the technique in \cite{LapidothTelatar98} by using
Lemmas \ref{lemma:Fsuperadditive} and \ref{lemma:boundE} to bound
the error probability independent of both $s_0$ and $\theta$. For a
given rate $R < C$, let $\epsilon=(C-R)/2$ and pick $m$ in such a
way that $\hat{C}_m \geq R + \epsilon$. Then
\begin{equation}
\max_{Q_{X^m||Z^{m-1}}} \inf_{s_0,\theta} \frac{1}{m} \mathcal
I(Q_{X^m||Z^{m-1}};P_{Y^m||X^m,s_0,\theta}) - \frac{\log|{\cal
S}|}{m} \geq R + \epsilon .
\end{equation}
Let $Q_m^* \in {\cal P}({\cal X}^m||{\cal Z}^{m-1})$ be the input
distribution that achieves the supremum in $\hat{C}_m$, i.e.,
\begin{equation}
\inf_{s_0,\theta} \frac{1}{m} \mathcal I(Q_m^*
;P_{Y^m||X^m,s_0,\theta}) - \frac{\log|{\cal S}|}{m} \geq R +
\epsilon \label{eqn:Qm*}
\end{equation}
Next, we use $Q_m^*$ to define a distribution $Q_{Nm} \in {\cal
P}({\cal X}^{Nm}||{\cal Z}^{Nm-1})$ for a sequence of length $Nm$,
$N \geq 1$, as follows.
\begin{eqnarray}
Q(x^{Nm}||z^{Nm{-}1}) & \!\!\!\!\!\triangleq \!\!\!&
Q_m^*(x_1^{m}||z_1^{m{-}1}) \times
Q_m^*(x_{m{+}1}^{2m}||z_{m{+}1}^{2m{-}1}) \times \hdots \times
Q_m^*(x_{(N{-}1)m{+}1}^{Nm}||z_{(N{-}1)m{+}1}^{Nm{-}1}) \\
 & \!\!\!\!\!=\!\!\! & \prod_{i=1}^{N} Q_m^*(x_{(i-1)m+1}^{im}||z_{(i-1)m+1}^{im-1})
\end{eqnarray}

For this new input distribution and sequence of length $Nm$, we can
bound the error exponent
\begin{equation}
F^{Nm}(\rho,Q_{Nm},\theta) -\rho R
\end{equation}
as shown below.
\begin{eqnarray}
  & \stackrel{(a)}
\geq &  F^{m}(\rho,Q_m^*,\theta) - \rho R \\
  & = & \min_{s_0} E_0(\rho,Q_m^*, s_0,\theta) - \rho \left(R + \frac{\log|{\cal S}|}{m}
 \right) \\
  & \stackrel{(b)} \geq &  \min_{s_0} \frac{1}{m} \rho
 I(Q_m^*;P_{Y^m||X^m,s_0,\theta}) - \frac{1}{2m} \rho^2
\left( \log(e|{\cal Y}^m|) \right)^2 - \rho \left(R +
\frac{\log|{\cal S}|}{m}
 \right) \\
  & \geq & \rho \left( \inf_{s_0,\theta} \frac{1}{m}
 I(Q_m^*;P_{Y^m||X^m,s_0,\theta}) - R -  \frac{\log|{\cal S}|}{m}
 \right) - \frac{1}{2m} \rho^2
\left( \log(e|{\cal Y}^m|) \right)^2 \\
  & \stackrel{(c)} \geq &  \rho \epsilon - \frac{1}{2m} \rho^2
\left( \log(e|{\cal Y}^m|) \right)^2
\end{eqnarray}
where $(a)$ is due to Lemma \ref{lemma:Fsuperadditive}, $(b)$
follows from Lemma \ref{lemma:boundE}, and $(c)$ follows from
(\ref{eqn:Qm*}). As in \cite{LapidothTelatar98}, we can maximize the
lower bound on the error exponent by setting
$\rho=\min(1,m\epsilon/\left( \log(e|{\cal Y}^m|) \right)^2)$. With
this choice of $\rho$ we have
\begin{equation}
F^{Nm}(\rho,Q_{Nm},\theta) -\rho R \geq \begin{cases} m
\epsilon^2/(2 \log(e|{\cal Y}|^m)^2) & \epsilon < \frac{1}{m}
(\log(e|{\cal Y}|^m) )^2
\\ \epsilon - \frac{1}{2m} \left( \log(e|{\cal Y}|^m) \right)^2 &
\mbox{otherwise}. \end{cases} \label{eqn:LowerBoundFNm}
\end{equation}
Theorem \ref{thrm:AchievabilityThetaKnown} follows by combining
(\ref{eqn:LowerBoundFNm}) with the result in Corollary
\ref{cor:boundProbErrorFE} (for blocklength $Nm$).

\section{Proof of Lemma \ref{lemma:CompoundFSCIsStronglySep}}
\label{app:CompoundFSCISStronglySep}

To prove the lemma, we must first establish two equalities relating
the channel causal conditioning distribution
$P(y^n||x^n,s_0,\theta)$ to the channel probability law
$P(y_i,s_i|x_i,s_{i-1},\theta)$. The following set of equalities
hold.
\begin{eqnarray}
P(y^n,x^n|s_0,\theta) & = & \sum_{s^n \in {\cal S}^n}
P(y^n,x^n,s^n|s_0,\theta) \\
 & \stackrel{(a)} = & \sum_{s^n \in {\cal S}^n} P(x^n||y^{n-1},s^{n-1},s_0,\theta)
 P(y^n,s^n||x^n,s_0,\theta) \\
 & \stackrel{(b)} = & \sum_{s^n \in {\cal S}^n} P(x^n||y^{n-1},s_0,\theta)
 P(y^n,s^n||x^n,s_0,\theta) \\
 & = & P(x^n||y^{n-1},s_0,\theta) \sum_{s^n \in {\cal S}^n}
 P(y^n,s^n||x^n,s_0,\theta) \label{eqn:ChainRuleCC}
\end{eqnarray}
where $(a)$ is due to \cite[Lemma 2]{Permuter06_feedback_submit} and
$(b)$ follows from our assumption that the input distribution $x^n$
does not depend on the state sequence $s^{n-1}$. By the chain rule
for causal conditioning \cite[Lemma 1]{Permuter06_feedback_submit},
(\ref{eqn:ChainRuleCC}) implies that
\begin{equation}
P(y^n||x^n,s_0,\theta) = \sum_{s^n \in {\cal S}^n}
P(y^n,s^n||x^n,s_0,\theta). \label{eqn:EquivFederLapidoth63}
\end{equation}
Also,
\begin{eqnarray}
P(y^n,s^n||x^n,s_0,\theta) & = & \prod_{i=1}^n
P(y_i,s_i|x^{i-1},y^{i-1},s^{i-1},\theta) \\
 & \stackrel{(c)} = & \prod_{i=1}^n P(y_i,s_i|x_i,s_{i-1},\theta)
 \label{eqn:EquivFederLapidoth64}
\end{eqnarray}
where $(c)$ follows from the definition of the compound finite-state
channel. Having established equations
(\ref{eqn:EquivFederLapidoth63}) and
(\ref{eqn:EquivFederLapidoth64}), Lemma
\ref{lemma:CompoundFSCIsStronglySep} follows immediately from
\cite[Lemma 12]{FederLapidoth98}, where the conditional probability
$P(y_i,s_i|x_i,s_{i-1},\theta)$ is quantized and the quantization
cells are represented by channels $\{\theta_1^{(n)}, \hdots,
\theta_{K(n)}^{(n)} \}$. The proof of our result differs only in
that the upper bound on the error exponents in the family is given
by $\mu = 1 + \log|\mathcal Y|$.

\section{Proof of Lemmas \ref{lemma:QkQm_inequality},
\ref{lemma:capacity_stationary_ergodic_Markovian} and
\ref{lemma:directed0iff}}  \label{app:LemmasForIFF}

The proof of Lemma \ref{lemma:QkQm_inequality} is based on an
identity that is given by Kim in \cite[eq. (9)]{Kim07_feedback}:
\begin{equation}
I(X^n\to Y^n)=\sum_{i=1}^n I(X_i;Y_i^n|X^{i-1},Y^{i-1})
\end{equation}

{\it Proof of Lemma \ref{lemma:QkQm_inequality}:} Using Kim's
identity we have
\begin{eqnarray}
I(X^n\to Y^n)&=&\sum_{i=1}^n I(X_i;Y_i^n|X^{i-1},Y^{i-1})\nonumber \\
&=&\sum_{i=1}^k I(X_i;Y_i^n|X^{i-1},Y^{i-1})+\sum_{i=k+1}^n I(X_i;Y_i^n|X^{i-1},Y^{i-1})\nonumber \\
&\geq&\sum_{i=1}^k I(X_i;Y_i^k|X^{i-1},Y^{i-1})+\sum_{i=k+1}^n I(X_i;Y_i^n|X^{i-1},Y^{i-1})\nonumber \\
&=&I(X^k\to Y^k)+\sum_{i=k+1}^n I(X_i;Y_i^n|X^{i-1},Y^{i-1}).
\end{eqnarray}
Now we bound the sum in the last equality,
\begin{eqnarray}
\sum_{i=k+1}^n I(X_i;Y_i^n|X^{i-1},Y^{i-1})&=&\sum_{i=k+1}^n H(X_i|X^{i-1},Y^{i-1})-H(X_i|X^{i-1},Y^{i-1},Y_i^n) \nonumber \\
&\stackrel{(a)} =&\sum_{i=k+1}^n H(X_i|X_{k+1}^{i-1},Y_{k+1}^{i-1})-H(X_i|X^{i-1},Y^{i-1},Y_i^n)\nonumber \\
&\geq& \sum_{i=k+1}^n H(X_i|X_{k+1}^{i-1},Y_{k+1}^{i-1})-H(X_i|X_{k+1}^{i-1},Y_{k+1}^{i-1},Y_i^n)\nonumber \\
&=&I(X_{k+1}^n\to Y_{k+1}^n)
\end{eqnarray}
where $(a)$ follows from the assumption that $Q(x^n||z^{n-1}) =
Q(x_1^k||z_1^{k-1})Q(x_{k+1}^{n}||z_{k+1}^{n-1})$. \hfill \QED

{\it Proof of Lemma
\ref{lemma:capacity_stationary_ergodic_Markovian}}: The proof
consists of two parts. In the first part we show that
$nC_n^{Markovian}$ is sup-additive and therefore $\lim_{n\to \infty}
C_n^{Markovian}=\sup_n C_n^{Markovian}$.  In the second part we
prove the capacity of the family of stationary and uniformly ergodic
Markovian channels by showing that
\begin{equation}
\lim_{n \rightarrow \infty} C_n=\lim_{n \rightarrow \infty}
C_n^{Markovian}.
\end{equation}
where $C_n$ is defined in (\ref{eqn:C_nDefinition}).

{\it First part:} We show that the sequence $C_n^{Markovian}$ is
sup-additive and therefore the limit exists. Let integers $k$ and
$m$ be such that $k+m=n$ and denote input distributions
$Q(x^n||z^{n-1}), Q(x_1^k||z_1^{k-1})$, and
$Q(x_{k+1}^{n}||z_{k+1}^{n-1})$ in shortened forms as $Q_n,Q_k$, and
$Q_m$. We have,
\begin{eqnarray}
nC_n^{Markovian}&=&\max_{Q_n} \inf_{\theta} I(X^n\to Y^n |\theta)\nonumber \\
&\stackrel{(a)}{\geq}&\max_{Q_kQ_m} \inf_{\theta} I(X^n\to Y^n |\theta)\nonumber \\
&\stackrel{(b)}{\geq}&\max_{Q_kQ_m} \inf_{\theta} \left[I(X^k\to Y^k |\theta)+I(X_{k+1}^n\to Y_{k+1}^n |\theta)\right]\nonumber \\
&\stackrel{}{\geq}&\max_{Q_kQ_m} \left[\inf_{\theta} I(X^k\to Y^k |\theta)+\inf_{\theta} I(X_{k+1}^n\to Y_{k+1}^n |\theta)\right]\nonumber \\
&\stackrel{=}{}&\max_{Q_k} \inf_{\theta} I(X^k\to Y^k |\theta)+\max_{Q_m}\inf_{\theta} I(X_{k+1}^n\to Y_{k+1}^n |\theta)\nonumber \\
&\stackrel{(c)}=&\max_{Q_k} \inf_{\theta} I(X^k\to Y^k |\theta)+\max_{Q(x^m||z^{m-1})}\inf_{\theta} I(X^m\to Y^m |\theta)\nonumber \\
&\stackrel{=}{}&kC_k^{Markovian}+mC_m^{Markovian},
\end{eqnarray}
where $(a)$ follows by restricting the maximization to causal
conditioning probabilities of the product form
$Q(x^n||z^{n-1})=Q(x_1^k||z_1^{k-1})Q(x_{k+1}^{n}||z_{k+1}^{n-1})$,
$(b)$ follows from Lemma \ref{lemma:QkQm_inequality}, and $(c)$
follows from stationarity of the channel.

{\it Second part:} We show that $\lim_{n \rightarrow \infty}
C_n=\lim_{n \rightarrow \infty} C_n^{Markovian}$. Due to Lemma 5 in
\cite{Permuter06_feedback_submit}, $| I(X^n\to Y^n|\theta) -
I(X^n\to Y^n|S_0,\theta)| \leq \log |\mathcal S|$, therefore it is
enough to prove that
\begin{equation}\label{eqn:diff_S0_s0}
\lim_{n \rightarrow \infty}\frac{1}{n} \left [
\max_{Q_{X^n||Z^{n-1}}} \inf_{ \theta}  I(X^n\to Y^n|S_0,\theta) -
\max_{Q_{X^n||Z^{n-1}}} \inf_{ \theta,s_0} I(X^n\to Y^n,
|s_0,\theta) \right ]=0.
\end{equation}
The difference in (\ref{eqn:diff_S0_s0}) is always positive, hence
it is enough to upper bound it by an expression that goes to zero as
$n\to \infty$. Again by Lemma 5 in \cite{Permuter06_feedback_submit}
we can bound the second term in (\ref{eqn:diff_S0_s0}),
\begin{eqnarray}\label{e_bound_directed}
\lefteqn{\max_{Q_{X^n||Z^{n-1}}} \inf_{ \theta,s_0} I(X^n\to Y^n,
|s_0,\theta)}\nonumber \\
&\geq& \max_{Q_{X^n||Z^{n-1}}} \inf_{ \theta,s_0} I(X^n\to Y^n,
|S_k,s_0,\theta) -\log |\mathcal S|\nonumber \\
&\stackrel{(a)}{\geq}& \max_{Q_{X_k^{n}||Z_k^{n-1}}} \inf_{
\theta,s_0} I(X_k^n\to
Y_k^n, |S_k,s_0,\theta) -\log |\mathcal S|,\nonumber \\
 &\stackrel{(b)}{=}&
\max_{Q_{X^{n-k}||Z^{n-k-1}}} \inf_{ \theta, s_{-k}} I(X^{n-k}\to
Y^{n-k}, |S_0,s_{-k},\theta) -\log |\mathcal S|,
\end{eqnarray}
where (a) holds for every $k>1$ and is due to Lemma
\ref{lemma:QkQm_inequality} and (b) holds by the stationarity of the
channel.
Hence, (\ref{e_bound_directed}) implies that we can bound the
difference,
\begin{eqnarray}\label{e_bound_diff3}
\lefteqn{\max_{Q_{X^n||Z^{n-1}}} \inf_{ \theta}  I(X^n\to
Y^n|S_0,\theta) - \max_{Q_{X^n||Z^{n-1}}} \inf_{ \theta,s_0}
I(X^n\to Y^n, |s_0,\theta)}\nonumber \\
&\stackrel{(a)}{\leq}& \left(k\log|\mathcal
Y|+\max_{Q_{X^{n-k}||Z^{n-k-1}}} \inf_{ \theta}  I(X^{n-k}\to
Y^{n-k}|S_0,\theta)\right)\nonumber
\\&&-\left(\max_{Q_{X^{n-k}||Z^{n-k-1}}} \inf_{ \theta,s_{-k}}
I(X_1^{n-k}\to Y^{n-k}, |S_0,s_{-k},\theta) -\log |\mathcal
S|\right), \nonumber \\
&\stackrel{(b)}{\leq}& k\log|\mathcal Y|+\epsilon (n-k)\log|\mathcal
Y| +\log |\mathcal S|.
\end{eqnarray}
Inequality (a) is due to the fact that $I(X^n\to Y^n)\leq
k\log|\mathcal Y|+I(X^{n-k}\to Y^{n-k})$ and due to
(\ref{e_bound_directed}). Inequality (b) holds since for a uniformly
ergodic family of channels,
$|P(s_0|s_{-k},\theta)-P(s_0|\theta)|\leq \epsilon$ for all $s_0\in
\mathcal S$ implies that for any input distribution
$Q_{X^{n-k}||Z^{n-k-1}}$ and any channel $\theta$,
$$ |I(X^{n-k}\to Y^{n-k}|\theta,S_0)-I(X_1^{n-k}\to Y^{n-k},
|S_0,s_{-k},\theta)|\leq \epsilon (n-k) \log |\mathcal Y|$$ After
dividing (\ref{e_bound_diff3}) by $n$, and since $\epsilon$ can be
arbitrarily small and $k$ is fixed for a given $\epsilon$, then
(\ref{eqn:diff_S0_s0}) holds.

\hfill \QED

{\it Proof of Lemma \ref{lemma:directed0iff}}: From the assumption
of the lemma we have %
\begin{equation}\label{eqn:directed_as_divergence}
\sum_{x^n,y^n}Q(x^n)P(y^n||x^n)\log
\frac{Q(x^n)P(y^n||x^n)}{P(y^n)Q(x^n)}=0.
\end{equation}
By assuming a uniform input distribution, $Q(x^n)=\frac{1}{|\mathcal
X|^n}$ and by using the fact that if the Kullback Leibler divergence
$D(p||q)\triangleq \sum_{x\in \mathcal X} p(x)\log
\frac{p(x)}{q(x)}$ is zero, then $p(x)=q(x)$ for all $x\in \mathcal
X$, we get that (\ref{eqn:directed_as_divergence}) implies that
$P(y^n||x^n)=P(y^n)$ for all $x^n\in \mathcal X^n,y^n\in \mathcal
Y^n$. It follows that
\begin{eqnarray}
\max_{Q_{X^n||Y^{n-1}}} I(X^n\to Y^n)
 & = & \max_{Q_{X^n||Y^{n-1}}} E\left[ \log
\frac{P(Y^n||X^n)}{P(Y^n)}\right] \\
 & = & \max_{Q_{X^n||Y^{n-1}}} E[0]=0.
\end{eqnarray}
\hfill \QED

\section{Proof of Lemmas \ref{l_DoNotLooseMuch} and
\ref{l_estimate_channel}} \label{app:LemmasForMemorylessCompound}

{\it Proof of Lemma \ref{l_DoNotLooseMuch}:} The proof is based on
the fact that $\mathcal I(Q_X,P_{Y|X})$ is uniformly continuous in
$P_{Y|X}$, namely for any $Q_X$,
\begin{equation}\label{e_uniform_cont_mutual}
|\mathcal I(Q_X,P_{Y|X,\theta_1})-\mathcal
I(Q_X,P_{Y|X,\theta_2})|\leq \tau(\Delta) ,
\end{equation}
where $\tau(\Delta)\to 0$ as $\Delta\to 0$ (The uniform continuity
of mutual information is a straightforward result of the uniform
continuity of entropy \cite[Theorem 2.7]{Csiszar81}). We have,
\begin{eqnarray}\label{e_bounding_diff}
\lefteqn{|\mathcal I(Q_X^{2},P_{Y|X,\theta_1})-\mathcal
I(Q_X^{1},P_{Y|X,\theta_1})|}\nonumber \\
&=& |\mathcal I(Q_X^{2},P_{Y|X,\theta_1})-\mathcal
I(Q_X^{2},P_{Y|X,\theta_2})+\mathcal
I(Q_X^{2},P_{Y|X,\theta_2})-\mathcal
I(Q_X^{1},P_{Y|X,\theta_1})|\nonumber\\
&\leq& \tau(\Delta)+ |\mathcal I(Q_X^{2},P_{Y|X,\theta_2})-\mathcal
I(Q_X^{1},P_{Y|X,\theta_1})|,
\end{eqnarray}
where the last inequality is due to (\ref{e_uniform_cont_mutual}).
We conclude the proof by bounding the last term in
(\ref{e_bounding_diff}) by $\tau(\Delta)$, which implies that if we
let $\eta(\Delta)=2\tau(\Delta)$ then (\ref{e_l_diff}) holds.
\begin{eqnarray}
\lefteqn{\mathcal I(Q_X^{2},P_{Y|X,\theta_2})-\mathcal
I(Q_X^{1},P_{Y|X,\theta_1})}\nonumber \\
&\leq&\mathcal I(Q_X^{2},P_{Y|X,\theta_2})-\mathcal
I(Q_X^{2},P_{Y|X,\theta_1}) \nonumber \\
&\leq& \tau(\Delta).
\end{eqnarray}
Similarly, we have ${\mathcal I(Q_X^{1},P_{Y|X,\theta_1})-\mathcal
I(Q_X^{2},P_{Y|X,\theta_2})} \leq \tau(\Delta)$, and therefore
\begin{eqnarray}\label{e_bounding_diff2}
|\mathcal I(Q_X^{2},P_{Y|X,\theta_2})-\mathcal
I(Q_X^{1},P_{Y|X,\theta_1})| &\leq& \tau(\Delta).
\end{eqnarray}
\hfill \QED

{\it Proof of Lemma \ref{l_estimate_channel}:} The channel $
P_{Y|X,\hat \theta}$ is chosen by finding the conditional empirical
distribution induced by an input sequence consisting of
$\frac{M}{|\mathcal X|}$ copies of each symbol of the alphabet
$\mathcal X$. We estimate the conditional distribution $P_{Y|a}$
separately for each $a\in \mathcal X$. We insert $x=a$ for
$m=\frac{M}{|\mathcal X|}$ uses of the channel and we estimate the
channel distribution when the input is $x=a$ as the type of the
output which is denoted as $P_{Y^m|a}$. From Sanov's theorem (cf.
\cite[Theorem 12.4.1]{CovThom}) we have that the probability that
type $P_{Y^m|a}$ will be at $L_1$-distance larger than
$\epsilon_1=\frac{\epsilon}{|\mathcal X|}$ from $P_{Y|a}$ is upper
bounded by
\begin{equation}
\Pr\{||P_{Y^m|a}-P_{Y|a}||_1\geq \epsilon_1\}\leq (m+1)^{|\mathcal
Y|}\text {exp}(-m\min_{P_{Y}: ||P_{Y}-P_{Y|a})||_1\geq \epsilon_1}
D(P_{Y}||P_{Y|a}),
\end{equation}
where $D(P_{Y}||P_{Y|a})=\sum_{y\in \mathcal Y} P_{Y}(y)\log
\frac{P_{Y}(y)}{P_{Y|a}(y|a)}$ denotes the divergence between the
two distributions. Using Pinsker's inequality \cite[Lemma
12.6.1]{CovThom} we have that
\begin{equation}
\min_{P_{Y}:||P_{Y}-P_{Y|a})||_1\geq \epsilon_1}
D(P_{Y}||P_{Y|a})\geq \frac{\epsilon_1^2}{2}
\end{equation}
and therefore,
\begin{equation}
\Pr\{||P_{Y^m}-P_{Y|a}||_1\geq \epsilon_1\}\leq (m+1)^{|\mathcal
Y|}\exp\left(-m\frac{\epsilon_1^2}{2}\right)
\end{equation}
The term $(m+1)^{|\mathcal Y|}\text {exp}(-m\frac{\epsilon_1^2}{2})$
goes to zero as $m$ goes to infinity for $\epsilon_1>0$ and
therefore, for any $\frac{\delta}{|\mathcal X|}>0$ we can find an
$m$ such that $(m+1)^{|\mathcal Y|}\text
{exp}(-m\frac{\epsilon_1^2}{2})\leq \frac{\delta}{|\mathcal X|}$.
Finally we have,
\begin{equation}
\Pr\{\Delta > \epsilon \}\leq \Pr\left \{\bigcup_{a\in \mathcal
X}||P_{Y|a,\hat \theta}-P_{Y|a}||_1> \frac{\epsilon}{|\mathcal
X|}\right\} \leq {|\mathcal X|} \frac{\delta}{|\mathcal X|}
\end{equation}
where the inequality on the right is due to the union bound. \hfill
\QED

\section*{Acknowledgments}
The authors would like to thank their advisors - Anthony Ephremides
and Tsachy Weissman - as well as Prakash Narayan for useful
discussions on this topic and Andrea Goldsmith for organizing the
Roundtable Research Discussion at ISIT06 which led to the conception
of this work.


\bibliographystyle{IEEEtran}

\begin{thebibliography}{10}
\providecommand{\url}[1]{#1} \csname url@rmstyle\endcsname
\providecommand{\newblock}{\relax} \providecommand{\bibinfo}[2]{#2}
\providecommand\BIBentrySTDinterwordspacing{\spaceskip=0pt\relax}
\providecommand\BIBentryALTinterwordstretchfactor{4}
\providecommand\BIBentryALTinterwordspacing{\spaceskip=\fontdimen2\font
plus \BIBentryALTinterwordstretchfactor\fontdimen3\font minus
  \fontdimen4\font\relax}
\providecommand\BIBforeignlanguage[2]{{%
\expandafter\ifx\csname l@#1\endcsname\relax
\typeout{** WARNING: IEEEtran.bst: No hyphenation pattern has been}%
\typeout{** loaded for the language `#1'. Using the pattern for}%
\typeout{** the default language instead.}%
\else \language=\csname l@#1\endcsname \fi #2}}

\bibitem{Blackwell59Compound}
D.~Blackwell, L.~Breiman, and A.~Thomasian, ``The capacity of a
class of
  channels,'' \emph{Ann. Math. Statist}, vol.~30, p. 1229, 1959.

\bibitem{Wolfowitz59}
J.~Wolfowitz, ``Simultaneous channels,'' \emph{Archive for Rational
Mechanics
  and Analysis}, vol.~4, pp. 371--386, 1959.

\bibitem{Wolfowitz64}
------, \emph{Coding Theorems of Information Theory}, 2nd~ed.\hskip 1em plus
  0.5em minus 0.4em\relax Springer, 1964.

\bibitem{LapidothTelatar98}
A.~Lapidoth and I.~E. Telatar, ``The compound channel capacity of a
class of
  finite-state channels,'' \emph{IEEE Trans. Inform. Theory}, vol.~44, pp.
  973--983, May 1998.

\bibitem{Lapidoth98Narayan}
A.~Lapidoth and P.~Narayan, ``Reliable communication under channel
  uncertainty,'' \emph{IEEE Trans. Inform. Theory}, vol.~44, 1998.

\bibitem{Gallager68}
R.~G. Gallager, \emph{Information Theory and Reliable
Communication}.\hskip 1em
  plus 0.5em minus 0.4em\relax New York: Wiley, 1968.

\bibitem{goldsmith96capacity}
A.~J. Goldsmith and P.~P. Varaiya, ``Capacity, mutual information,
and coding
  for finite-state {M}arkov channels,'' \emph{IEEE Trans. on Info. Theory},
  vol.~42, pp. 868--886, 1996.

\bibitem{Mushkin89}
M.~Mushkin and I.~Bar-David, ``Capacity and coding for the {G}ilbert
{E}lliot
  channel,'' \emph{IEEE Trans. Inform. Theory}, vol.~35, pp. 1277--1290, 1989.

\bibitem{FederLapidoth98}
M.~Feder and A.~Lapidoth, ``Universal decoding for channels with
memory,''
  \emph{IEEE Trans. Inform. Theory}, vol.~44, no.~5, September 1998.

\bibitem{Massey90}
J.~Massey, ``Causality, feedback and directed information,''
\emph{Proc. Int.
  Symp. Information Theory Application (ISITA-90)}, pp. 303--305, 1990.

\bibitem{Kramer03}
G.~Kramer, ``Capacity results for the discrete memoryless network,''
\emph{IEEE
  Trans. Inform. Theory}, vol.~49, pp. 4--21, 2003.

\bibitem{Tatikonda00}
S.~Tatikonda, ``Control under communication constraints,''
\emph{Ph.D.
  disertation, MIT, Cambridge, MA}, 2000.

\bibitem{Tatikonda06}
\BIBentryALTinterwordspacing S.~Tatikonda and S.~Mitter, ``The
capacity of channels with feedback,''
  September 2006. [Online]. Available:
  \url{http://arxiv.org/PS_cache/cs/pdf/0609/0609139.pdf}
\BIBentrySTDinterwordspacing

\bibitem{Permuter06_feedback_submit}
H.~H. Permuter, T.~Weissman, and A.~J. Goldsmith, ``Finite state
channels with
  time-invariant deterministic feedback,'' Sep 2006, submitted to IEEE Trans.
  Inform. Theory. Availble at http://arxiv.org/abs/cs/0608070v1.

\bibitem{Kim07_feedback}
Y.~Kim, ``A coding theorem for a class of stationary channels with
feedback,''
  Jan 2007, submitted to IEEE Trans. Inform. Theory. Availble at
  arxiv.org/cs.IT/0701041.

\bibitem{TchamkertenTelatar06}
A.~Tchamkerten and I.~Telatar, ``Variable length coding over an
unknown
  channel,'' \emph{IEEE Trans. Inform. Theory}, vol.~52, no.~5, pp. 2126--2145,
  2006.

\bibitem{Alajaji95}
F.~Alajaji, ``Feedback does not increase the capacity of discrete
channels with
  additive noise,'' \emph{IEEE Trans. Inform. Theory}, vol.~41, pp. 546--549,
  March 1995.

\bibitem{Pinsker_feedback_double}
M.~Pinsker, talk delivered at the {S}oviet {I}nformation {T}heory
{M}eeting (no
  abstract published), 1969.

\bibitem{Ebert}
P.~Ebert, ``{The capacity of the Gaussian channel with feedback },''
\emph{Bell
  Syst. Tech. J.}, pp. 1705--1712, 1970.

\bibitem{Cover89}
T.~M. Cover and S.~Pombra, ``Gaussian feedback capacity,''
\emph{IEEE Trans.
  Inform. Theory}, vol.~35, no.~1, pp. 37--43, 1989.

\bibitem{Tchamkerten07}
A.~Tchamkerten, private communication, 2007.

\bibitem{Csiszar81}
I.~Csisz{\'a}r and J.~K{\"o}rner, \emph{Information Theory: Coding
Theorems for
  Discrete Memoryless Systems}.\hskip 1em plus 0.5em minus 0.4em\relax New
  York: Academic, 1981.

\bibitem{CovThom}
T.~Cover and J.~A. Thomas, \emph{Elements of Information
Theory}.\hskip 1em
  plus 0.5em minus 0.4em\relax Wiley, 1991.

\end{thebibliography}

\end{document}